\documentclass{amsart}
\usepackage{latexsym}
\usepackage{amssymb}
\usepackage{amsmath}
\usepackage{amsfonts}
   
\def\b{\textbf{b}}

\def\R{\mathbb{R}}

\newcommand{\bmat}{\left[\begin{matrix}}
\newcommand{\emat}{\end{matrix}\right]}
\usepackage[latin1]{inputenc} 
\usepackage{amsthm}
\usepackage{amsxtra}
\usepackage{amscd}
\usepackage{setspace}   
\usepackage{mathrsfs}   

\newtheorem{theorem}{Theorem}
\newtheorem{proposition}[theorem]{Proposition}

\newtheorem{conjecture}[theorem]{Conjecture}

\theoremstyle{remark}
\newtheorem*{remark}{Remark}

\theoremstyle{definition}

\newcommand{\Z}{\mathbb{Z}}

\newcommand{\vb}{{\mathbf b}}

\usepackage{tikz}
\usetikzlibrary{arrows}

\usepackage{algorithm}
\usepackage{algpseudocode}
\usepackage{enumerate}
\usepackage{graphicx}

    \title{A physical study of the LLL algorithm}
    \author{Jintai Ding, Seungki Kim, Tsuyoshi Takagi, Yuntao Wang, Bo-yin Yang}
   
\begin{document}

\begin{abstract}

This paper presents a study of the LLL algorithm from the perspective of statistical physics. Based on our experimental and theoretical results, we suggest that interpreting LLL as a sandpile model may help understand much of its mysterious behavior. In the language of physics, our work presents evidence that LLL and certain 1-d sandpile models with simpler toppling rules belong to the same universality class.

This paper consists of three parts. First, we introduce sandpile models whose statistics imitate those of LLL with compelling accuracy, which leads to the idea that there must exist a meaningful connection between the two. Indeed, on those sandpile models, we are able to prove the analogues of some of the most desired statements for LLL, such as the existence of the gap between the theoretical and the experimental RHF bounds. Furthermore, we test the formulas from finite-size scaling theory (FSS) against the LLL algorithm itself, and find that they are in excellent agreement. This in particular explains and refines the geometric series assumption (GSA), and allows one to extrapolate various quantities of interest to the dimension limit. In particular, we obtain the estimate that the empirical average RHF converges to $\approx 1.02265$ as the dimension goes to infinity.

\end{abstract}

\maketitle

\section{Introduction}

\subsection{The mysteries of LLL}

The LLL algorithm \cite{LLL82} is one of the most celebrated algorithmic inventions of the twentieth century, with countless applications to pure and computational number theory, computational science, and cryptography. It is also the most fundamental of lattice reduction algorithms, in that nearly all known reduction algorithms are generalizations of LLL in some sense, and they also utilize LLL as their subroutine. (We refer the reader to \cite{NV10} for a thorough survey on LLL and these related topics.) Thus it is rather curious that many of the salient features of LLL in practice is left totally unexplained, not even in a heuristic, speculative sense, even to this day.

The most well-known among the mysteries of LLL is the gap between its worst-case root Hermite factor(RHF) and the observed average-case, as documented in Nguyen and Stehl\'e \cite{NS06}. It is a theorem from the original LLL paper \cite{LLL82} that the shortest vector of an LLL-reduced basis (in the theoretical sense) in dimension $n$, with its determinant normalized to $1$, has length at most $(4/3)^{\frac{n-1}{4}} \approx 1.075^n$, whereas in practice one almost always observes $\approx 1.02^n$, regardless of the way in which the input is sampled. This is a strange phenomenon in the light of the works of Kim \cite{K15} and Kim and Venkatesh \cite{KV17}, which provide experimental and theoretical evidence that, for almost every lattice, nearly all of its LLL bases have RHF close to the worst bound. It is as though the algorithm is consciously dodging those plethora of inferior bases every time it is run. This leads to the suspicion that LLL must be operating in a complex manner that belies the simplicity of its code.

There are also many other LLL phenomena that remain unaccounted for. One is the geometric series assumption (GSA), originally proposed by Schnorr \cite{S03}, and its partial failure at the boundaries, both of which are observed in other blockwise reduction algorithms as well e.g. BKZ \cite{SE94}. Despite being an indispensable component of numerous cryptanalyses of lattice-based systems (e.g. see \cite{DY17}, \cite{BSW18}), the current understanding of GSA is not much better than that of the RHF gap problem above: not even a heuristic explanation, or a precise formulation, only vague empirical observations. There are also questions raised regarding the time complexity of LLL. Nguyen and Stehl\'e \cite{NS06} suggest that, in most practical situations, the average time complexity is much lower than the worst-case, suggesting that there may be the average-worst case gap phenomenon here as well. The complexity of the optimal LLL algorithm --- i.e. where the parameter $\delta$ equals $1$ --- is not proven to be polynomial-time, although observations suggest that it is (see Akhavi \cite{A00} and references therein).

This lack of understanding of the practical behavior of LLL --- and reduction algorithms in general --- may incur a hefty price, especially when it comes to cryptographic applications. To put it somewhat bluntly: simply by running LLL, we managed to ``improve'' the RHF of LLL from $1.075$ to $1.02$; what keeps one from entertaining the possibility that a cheap trick might improve it further to, say, $1.005$, and thereby cripple all lattice-based cryptosystems? As unrealistic --- and perhaps even outrageous --- as this may sound, our current understanding of reduction algorithms is severely unequipped to address this question.

\subsection{This paper}

The theme of the present paper is that statistical physics may enable a scientific approach to the empirical behavior of the LLL algorithm, by studying it as a kind of a sandpile model. As demonstrated throughout this paper, for each LLL phenomenon, there is a corresponding sandpile phenomenon, most of which are either already familiar to physicists or captured by well-known methods in physics. Some aspects of our work seem to present challenges to physics, and we hope those will motivate rich and fruitful interdisciplinary interactions revolving around the LLL algorithm, and lattice reduction algorithms in general.


In Section 2, we justify this perspective by presenting stochastic sandpile models that are both impressively close to LLL and mathematically accessible. Specifically, we propose two models of LLL, which we name \emph{LLL-SP} and \emph{SSP} respectively. LLL-SP (Algorithm \ref{alg:lllsp} below) is a non-Abelian stochastic model that exhibits nearly identical quantitative behavior to that of LLL in numerous aspects, both in terms of their output statistics such as the distribution of RHF, and their dynamics. This provides compelling evidence that the two algorithms operate under the same principles, or put it formally, that they are in the same universality class. SSP (Algorithm \ref{alg:ssp}) is an Abelian stochastic model that is mathematically far more tractable than LLL-SP, and still imitates the most important aspects of the output statistics of LLL.

In Sections 3 and 4, we prove on these models some of the most desired statements regarding LLL. On the RHF gap phenomenon, we have the following

\begin{theorem} \label{thm:intro_ssprhf}
In all sufficiently large system sizes (which corresponds to the lattice dimensions for LLL), there exists a gap between the worst-case and the average-case RHFs of SSP.
\end{theorem}

Theorem \ref{prop:ssprhf} below provides a more precise quantitative statement, after the necessary definitions are set up. We mention that the mathematical study of SSP and the proof of this theorem are announced in the companion paper \cite{Kprep}, separated from the present paper in order for consideration in a purely physical context. Hence Section 3, where we introduce Theorem \ref{thm:intro_ssprhf}, is expository, included for the completeness of the presentation of our perspective on LLL. We expect that a key idea in the proof of Theorem \ref{thm:intro_ssprhf} can be extended to yield the same result for LLL-SP; see Conjecture \ref{conj:mass}.

We are able to prove some fairly strong statements regarding the time complexity of LLL-SP (which also applies to SSP):

\begin{theorem} \label{thm:intro_time}
Choose an input basis $\{\mathbf{b}_1, \ldots, \mathbf{b}_n\} \subseteq \R^n$, and let $E = n^2 \log \max_i\|\mathbf{b}_i\|$. Then
\begin{itemize}
\item (Lower bound on complexity) There exists a constant $C$ such that, with probability $1 - CE^{-1/2}$, LLL-SP takes at least $E/4$ swaps to terminate.
\item (Polynomial-time complexity of the optimal LLL) With probability $1 - \eta$, the optimal LLL-SP --- that is, with the maximal $\delta$ parameter --- terminates within $O_\eta(E)$ swaps.
\end{itemize}
\end{theorem}

See Theorems \ref{thm:lowertime} and \ref{thm:optimal} for precise statements. The lower bound is of particular interest from the cryptographic perspective, since it sets a certain limit on the strength of lattice reduction algorithms. We expect that this result is also valid for LLL assuming a certain conjecture on its dynamical property that is well-supported by our experiments; see Conjecture \ref{conj:mu} below.

In Section 5, we further develop the connection between LLL and sandpile models by ``applying'' finite-size scaling theory (FSS) to LLL. FSS is a theory in physics that studies critical phase transitions, such as water freezing into ice, and metals being magnetized. Although there is no critical phenomenon to discuss for LLL, the analogy with sandpile models motivates us to investigate if some observables in LLL scale with dimension in a similar way to what is seen in physics in finite-size scaling theory of critical phenomena.

Denote by $y_n$ the natural log of the ``average RHF'' of LLL in dimension $n$, and $y_\infty := \lim_{n \rightarrow \infty} y_n$. Also, for a (LLL-reduced) basis $\mathcal{B} = \{\mathbf{b}_1, \ldots, \mathbf{b}_n\}$ and its Gram-Schmidt orthogonalization $\{\mathbf{b}_1^*, \ldots, \mathbf{b}_n^*\}$, write $r(i) = \log \|\mathbf{b}_i^*\|/\|\mathbf{b}_{i+1}^*\|$. Then the formulas from FSS that would normally apply to (Abelian) sandpiles translate to the following for LLL: there exists a single constant $\sigma$ such that

\begin{enumerate}[(i)]
\item $y_\infty = y_n + \frac{D}{n^\sigma} + \mbox{(smaller errors)}$, for some constant $D$.
\item $\mathrm{Var}(y_n) \sim n^{-2\sigma}$.
\item $2y_\infty - \mathbb{E}(r(i)) \sim i^{-\sigma} \mbox{ or } (n-i)^{-\sigma}$, depending on whether $i$ is near $1$ or $n-1$.
\end{enumerate}

All three statements are clearly interesting: (i) and (ii) are self-explanatory, and (iii) provides the correct formulation of the GSA (which says that $r(i)$ are nearly constant) and its partial failures near the boundaries.
Our data on dimensions up to $300$ --- summarized in Tables \ref{table:1st} and \ref{table:2nd}, and Figures \ref{fig:1st}-\ref{fig:3rd_b} below --- fit robustly with all of the above formulas with $\sigma \approx 0.75$. Accordingly, we obtain a numerical estimate
\begin{equation}\label{eq:FSSprediction}
 \mbox{(average RHF of LLL)} \rightarrow 1.02265\ldots, \mbox{ as $n \rightarrow \infty$.}
\end{equation}

It may be of interest that Grassberger, Dhar, and Mohanty \cite{GDM16} numerically obtained the same value of $\sigma \approx 0.75$ for a sandpile model with a very different toppling rule. In physics, different systems with the same critical exponents (such as $\sigma$ here) that govern their behavior in the system size limit are said to belong to the same \emph{universality class}. It is expected that there exist not too many distinct universality classes.

There exists some subtlety regarding (iii), arising from the fact that LLL is non-Abelian as a sandpile model. It does hold on one end with $\sigma \approx 0.75$ for the first 8-10 values of $i$, but on the other end, it holds with a different exponent $\approx 1.05$. At this point, we do not know how to explain this phenomenon in a satisfactory manner; it could be the size of our data --- which is quite large from the lattice reduction perspective, but tiny from the physical one --- or the authors' shortcomings in physics. At the very least, we obtain a neat extrapolation of $\mathbb{E}(r(i))$ on both ends, which has been of some recent cryptographic interest (see \cite{BSW18}, \cite{DY17}).



\subsection{Comparison with previous works}

This paper is not the first to compare LLL, and blockwise reduction algorithms in general, to a sandpile model. The formal similarity seems to have been first noticed in Madritsch and Vall\'ee \cite{MV10} --- see also Vall\'ee \cite{V16}. This idea was and is being more vigorously applied to the simulation of BKZ, the algorithm used in practice to challenge lattice-based cryptosystems that may be viewed as a generalization of LLL. We refer the readers to \cite{CN11}, \cite{HPS11}, and the more recent \cite{BSW18} for examples.

The present work most importantly differs in motivation from the above-mentioned works, and other related works in the cryptographic literature. In cryptography, often the goal is to craft what is called a \emph{simulator} of BKZ, an algorithm of very small temporal and spatial complexity that aids the practitioners in predicting the outcome of BKZ, with a particular interest in the RHF and the output profile. On the other hand, our goal is to search for a scientific theory that matches the observed behavior of LLL. It is one of our hopes that our work serves as a contribution to the construction of a better simulator, but we do not claim to be part of that competition.

This difference in our motivation is what leads us to investigate LLL in ways that have not been tried in the previous works, which are nearly exclusively focused on cryptographic applications. We subject our models to far more severe challenges --- running tens of thousands of tests, applying tweaks, comparing more observables than just the RHF --- than is done for the simulators. We do come up with a high-quality simulator of LLL as a result, yet that is the bare minimum necessity, not a sufficiency, to convince anyone that LLL may be governed by the laws of statistical physics, like the sandpile models are. Furthermore, adopting the well-developed ideas of physics such as the operator algebra method (Sections 3 and 4), and finite-size scaling theory (Section 5), we question some of the statements that have often been taken for granted, such as whether the number $1.02$ is not a mere anomaly of the small dimensions, and whether the GSA is really the ideal description of the output shape of LLL. 

We again stress that we are not pitting our work against the literature on BKZ simulators, and ask the reader to avoid the mistake of the same kind. Rather, we hope our work to be understood as an attempt to see LLL under a different light. Yes, LLL has been viewed as a sandpile model in the sense of an algorithm, but it has never been viewed as a sandpile model in the sense of an object subject to the principles of statistical mechanics. In that aspect our work is the first of its kind.


\subsection{Assumptions and notations}

In Sections 2-4, instead of the original LLL reduction from \cite{LLL82}, we work with its Siegel variant, a slight simplification of LLL. The Siegel reduction shares with LLL all the same qualitative features, but easier to handle theoretically, making it a reasonable starting point for our study. However, in Section 5 (the section on FSS), we revert to the original LLL, since it would be more interesting to extrapolate its RHF than that of the Siegel variant. Either way, our numerous smaller experiments suggest that the choice of LLL or Siegel affect the outcomes marginally at most.

The integer $n$ always represents the dimension of the relevant Euclidean space. Our lattices in $\mathbb{R}^n$ always have full rank. A basis $\mathcal{B}$, besides its usual definition, is an \emph{ordered} set, and we refer to its $i$-th element as $\vb_i$. Denote by $\vb_i^*$ the component of $\vb_i$ orthogonal to all vectors preceding it, i.e. $\vb_1, \ldots, \vb_{i-1}$. Also, for $i > j$, define $\mu_{i,j} := \langle \vb_i, \vb_j^* \rangle / \langle \vb_j^*, \vb_j^* \rangle$. Thus the following equality holds in general:
\begin{equation} \label{eq:gso_decomp}
\vb_i = \vb_i^* + \sum_{j=1}^{i-1} \mu_{i,j}\vb_j^*.
\end{equation}
We say $\mathcal B$ is \emph{size-reduced} if all $|\mu_{i,j}| \leq 0.5$. One can \emph{size-reduce} any basis $\mathcal B$, i.e. turn it into a size-reduced basis, by the following simple algorithm: for $j = n-1$ to $1$, and for each $j < i \leq n$, add or subtract $\vb_{j}$ from $\vb_i$ repeatedly until $|\mu_{i,j}| \leq 0.5$ holds (in computations, one sometimes allows $\mu_{i,j}$ to be slightly greater than $0.5$ in order to avoid floating-point errors). One can check using \eqref{eq:gso_decomp} that this procedure indeed produces a size-reduced basis.

We will write for shorthand $\alpha_i := \|\vb_i^*\| / \|\vb_{i+1}^*\|$, and $Q_i := (\alpha_i^{-2} + \mu_{i+1,i}^2)^{-1/2}$. When discussing lattices, $r_i := \log \alpha_i$, and when discussing sandpiles, $r_i$ refers to the ``amount of sand'' at vertex $i$.

\subsection{Data for the experiments} 
The original codes for the experiment are made available on SK's website https://sites.google.com/view/seungki/home. For the data, please consult one of the authors --- the raw data is of several gigabytes in size.

\subsection{Acknowledgments}
JD and SK are partially supported by NSF CNS-2034176. BY is supported by Sinica Investigator Award AS-IA-109-M01, and Executive Yuan Project AS-KPQ-109-DSTCP. TT and YW are supported by JSPS KAKENHI Grant Number JP20K23322.

We are hugely indebted to Deepak Dhar, who patiently explained much of the underlying physics over a long period of time, and directed us to the relevant works in physics. We also thank Deepak Dhar (again), Nick Genise, Steve D. Miller, and Phong Nguyen for their careful reading and comments, and Shi Bai for his extensive help with parts of the experiments in Section 5.

\section{Modeling LLL by a sandpile}

\subsection{The LLL algorithm}
We briefly review the LLL algorithm; for details, we recommend \cite{LLL82}, in which it is first introduced, and also \cite{JS98} and \cite{NV10}. A pseudocode for the LLL algorithm is provided in Algorithm \ref{alg:lll}. In Line 3, we deliberately left the choice algorithm, that is, the method for choosing $k$, unprescribed. The standard choice is to choose the lowest $k$ satisfying the inequality.

\begin{algorithm}
\caption{The LLL algorithm (Siegel variant)}\label{alg:lll}
\begin{enumerate}[1.]
\item[0.] Input: a basis $\mathcal{B} = \{\vb_1, \ldots, \vb_n\}$ of $\mathbb{R}^n$, a parameter $\delta < 0.75$
\item while true, do:
\item \hspace{4mm} Size-reduce $\mathcal{B}$.
\item \hspace{4mm} (Lov\'asz test) choose a $k \in \{1, \ldots, n-1\}$ such that $\delta\|\vb_{k}^*\|^2 > \|\vb_{k+1}^*\|^2$
\item \hspace{4mm} if there is no such $k$, break
\item \hspace{4mm} swap $\vb_{k}$ and $\vb_{k+1}$ in $\mathcal{B}$
\item Output $\mathcal{B} = \{\vb_1, \ldots, \vb_n\}$, a $\delta$-reduced LLL basis.
\end{enumerate}
\end{algorithm}

\begin{proposition} \label{prop:post_swap}
After carrying out Step 5 in Algorithm \ref{alg:lll}, the following changes occur:
\begin{enumerate}[(i)]
\item $\alpha_{k-1}^{new} = Q_k\alpha_{k-1}$
\item $\alpha_k^{new} =  Q_k^{-2}\alpha_k$
\item $\alpha_{k+1}^{new} =  Q_k\alpha_{k+1}$
\item $\mu_{k, k-1}^{new} =  \mu_{k+1, k-1}$
\item $\mu_{k+1, k}^{new} =  Q_k^2\mu_{k+1,k}$
\item $\mu_{k+2, k+1}^{new} =  \mu_{k+2,k} - \mu_{k+2,k+1}\mu_{k+1,k}$
\item $\mu_{k,l}^{new} = \mu_{k+1,l}, \mu_{k+1,l}^{new} = \mu_{k,l}$ for $1 \leq l \leq k-1$
\item $\mu_{l,k}^{new} = \mu_{l,k+1} - \mu_{l,k+1}\mu_{k+1,k}\mu_{k+1,k}^{new} + \mu_{l,k}\mu_{k+1,k}^{new}$ for $l \geq k+2$
\item $\mu_{l,k+1}^{new} = \mu_{l,k} - \mu_{l,k+1}\mu_{k+1,k}$ for $l \geq k+2$
\end{enumerate}
and there are no other changes. The superscript ``new'' refers to the corresponding variable after the swap.
\end{proposition}
\begin{proof}
Straightforward calculations (see e,g, \cite{LLL82}).
\end{proof}

\subsection{Sandpile basics}

We also briefly review the basics of the sandpile models. For references, see Dhar \cite{D99}, \cite{D06} or Perkinson \cite{P14}.

A sandpile model is defined on a finite graph $\mathcal{G}$, with one distinguished vertex called the \emph{sink}. In the present paper, we only concern ourselves with the cycle graph, say $A_n$, consisting of vertices $\{v_1, \ldots, v_n\}$ and one unoriented edge for each adjacent pair $v_i$ and $v_{i+1}$. We also consider $v_1$ and $v_n$ as adjacent. We designate $v_n$ as the sink.

A \emph{configuration} is a function $r : \{v_1, \ldots, v_n\} \rightarrow \mathbb{R}$. Just as reduction algorithms work with bases, sandpile models work with configurations. We write for short $r_i = r(v_i)$. One may think of $r_i$ as the amount or \emph{height} of the pile of sand placed on $v_i$.

\begin{figure}
\includegraphics[scale=1.2]{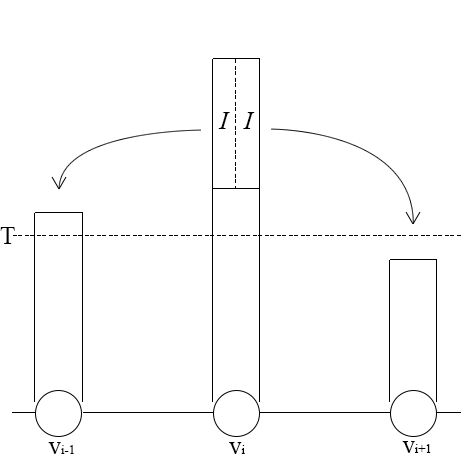}
\caption{An illustration of a (legal) toppling $T_i$.} \label{fig:toppling}
\end{figure}

Just as LLL computes a reduced basis by repeatedly swapping neighboring basis vectors, sandpiles compute a \emph{stable configuration} by repeated \emph{toppling.} Let $T, I \in \mathbb{R}_{>0}$. A configuration is \emph{stable} if $r_i \leq T$ for all $i \neq n$. A \emph{toppling operator} $T_i$ ($i \neq n$) replaces $r_i$ by $r_i - 2I$, and $r_{i-1}$ by $r_{i-1} + I$ and $r_{i+1}$ by $r_{i+1} + I$. An illustration is provided in Figure \ref{fig:toppling}. Applying $T_i$ when $r_i > T$ is called a \emph{legal toppling}. By repeatedly applying legal topplings, all excess ``sand'' will eventually be thrown away to the sink, and the process will terminate.

In our paper, $T$ --- \emph{threshold} --- will always be a fixed constant, but $I$ --- \emph{increment} --- could be a function of the current configuration, or a random variable, or both. If $I$ is independent of the configuration, we say the model is \emph{Abelian}, otherwise \emph{non-Abelian}. In Abelian models, the stable configuration reached is independent of the order of the legal topplings taken. This is not necessarily the case for non-Abelian models, as is demonstrated in Section 2.4 below.

If the increment $I$ is a random variable, we say the model is \emph{stochastic}. The (non-stochastic) Abelian sandpile theory is quite well-developed, with rich connections to other fields of mathematics --- see e.g. \cite{L10}. Other sandpile models are far less understood, especially the non-Abelian ones.

\subsection{The LLL sandpile model}

Motivated by Proposition \ref{prop:post_swap}, especially the formulas (i) -- (iii), we propose the following Algorithm \ref{alg:lllsp}, which we call the \emph{LLL sandpile model}, or LLL-SP for short. 

\begin{algorithm}
\caption{The LLL sandpile model (LLL-SP)}\label{alg:lllsp}
\begin{enumerate}[1.]
\item[0.] Input: $\alpha_1, \ldots, \alpha_{n} \in \mathbb{R}$, $\mu_{2,1}, \ldots, \mu_{n,n-1} \in [-0.5,0.5]$, a parameter $\delta < 0.75$
\item Let $r_i : = \log \alpha_i$, $\mu_i := \mu_{i+1,i}$ $T := -0.5\log \delta, Q_i := (\alpha_i^{-2} + \mu_{i+1,i}^2)^{-1/2}$.
\item while true, do:
\item \hspace{4mm} choose a $k \in \{1, \ldots, n-1\}$ such that $r_k > T$
\item \hspace{4mm} if there is no such k, break
\item \hspace{4mm} subtract $2\log Q_k$ from $r_k$
\item \hspace{4mm} add $\log Q_k$ to $r_{k-1}$ (if $k-1 \geq 1$) and $r_{k+1}$ (if $k+1 \leq n-1$)
\item \hspace{4mm} (re-)sample $\mu_{k-1}, \mu_k, \mu_{k+1}$ uniformly from $[-0.5,0.5]$
\item Output: real numbers $r_1, \ldots, r_{n-1} \leq T$
\end{enumerate}
\end{algorithm}

The only difference between LLL (Algorithm \ref{alg:lll}) and LLL-SP (Algorithm \ref{alg:lllsp}) lies in the way in which the $\mu$'s are replaced after each swap or topple. Our experimental results below demonstrate that this change hardly causes any difference in their behavior. A theoretical perspective is discussed at the end of this section.

The increment $I = \log Q_i = -\frac{1}{2}\log (e^{-2r_i} + \mu^2_i)$ is not as unnatural as it might seem --- see Figure \ref{fig:incr}. The dashed lines there represent the graph of
\begin{equation*}
I_\mu(r) = \begin{cases} r &\mbox{if $r > -\log \mu$} \\ -\log \mu &\mbox{otherwise.} \end{cases}
\end{equation*}
for comparison.
The decision to sample $\mu_i$'s uniformly is largely provisional, though some post hoc justification is provided in Figure \ref{fig:museq}. If desired, one could refine the model by adopting part of Proposition \ref{prop:post_swap} for updating $\mu_i$.

\begin{figure} 
\centering
\includegraphics[scale=0.2]{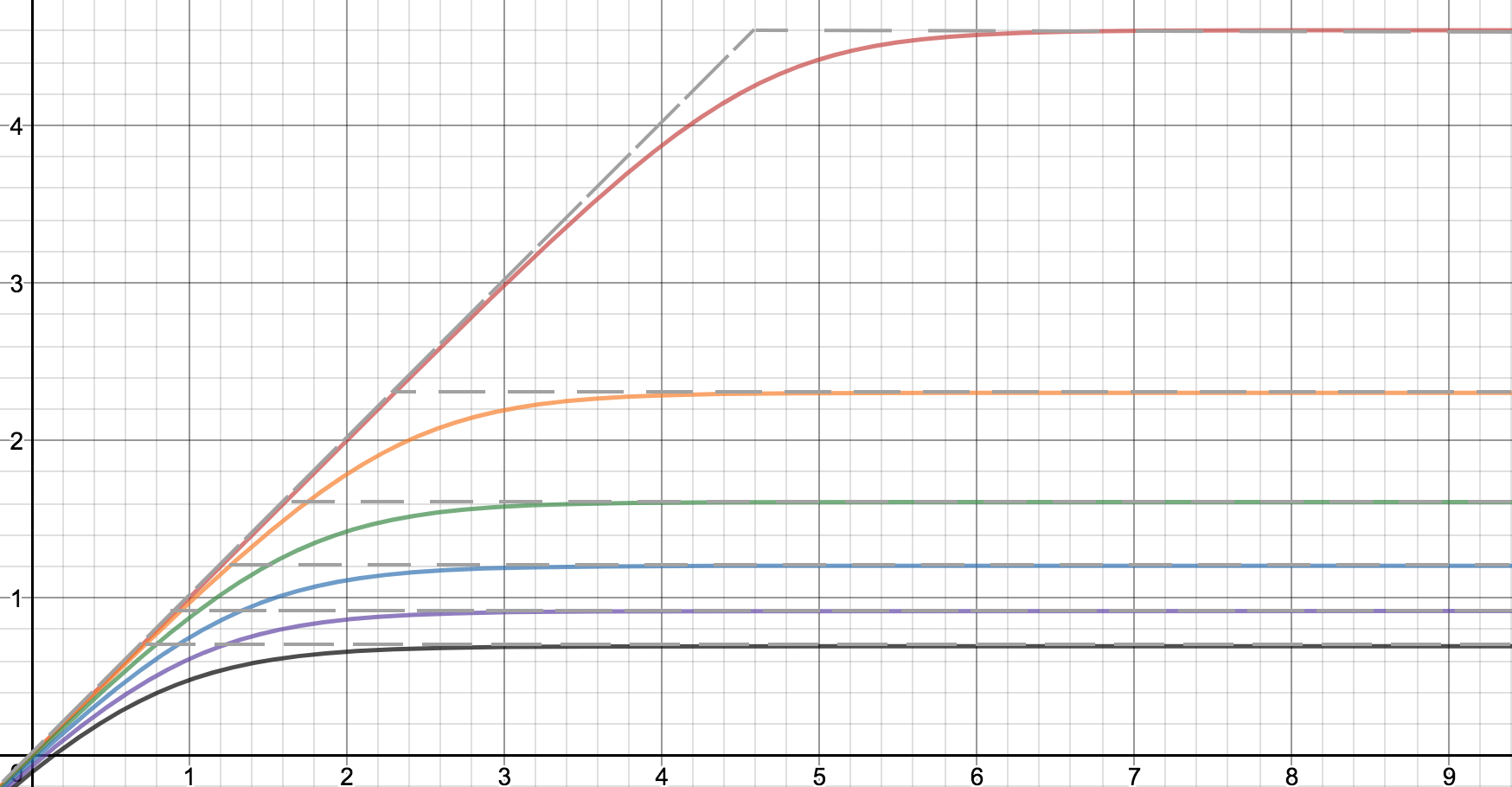}
\caption{Graphs of $\log Q_i$ as a function of $r_i$, for $\mu = 0.01, 0.1, 0.2, 0.3, 0.4, 0.5$, from top to bottom. The graph corresponding to $\mu = 0.5$ crosses the $x$-axis at $x = T \approx 0.1438$. } \label{fig:incr}
\end{figure}

\subsection{Numerical comparisons}


For each dimension $n = 80, 100, 120$, we ran LLL and LLL-SP 5,000 times with the same set of input bases of determinant $\approx 2^{10n}$, generated using the standard method suggested in Section 3 of \cite{NS06}. We used fpLLL \cite{FPLLL} for the LLL algorithm. We remind the reader that we have used the Siegel variant here.

In addition, we also ran the same experiment with the following two other choice algorithms, to see how they affect the outcome:
\begin{itemize}
\item \emph{random}: randomly and uniformly choose an index from those on which swapping/toppling is available, and swap/topple on that index.
\item \emph{greedy}: swap/topple on the index with the greatest increment $\log Q_k$.
\end{itemize}

Figure \ref{fig:output} shows the average shape of the output bases and configurations by LLL and LLL-SP.  One easily observes that the algorithms yield nearly indistinguishable outputs (except possibly for the greedy; see Remark below). In particular, since RHF can be computed directly from the $r_i$'s by the formula

\begin{equation} \label{eq:rhf}
\mbox{RHF} = \exp\left(\frac{1}{n^2}\sum_{i=1}^{n-1}(n-i)r_i\right),
\end{equation}
we expect both to yield about the same RHF. Indeed, Table \ref{table:rhf} and Figure \ref{fig:RHFdist} show that the RHF distribution of LLL and LLL-SP are in excellent agreement (again except for greedy, for which the average differs by $\approx 0.0011$).

\begin{figure} 
\centering
\includegraphics[scale=0.30]{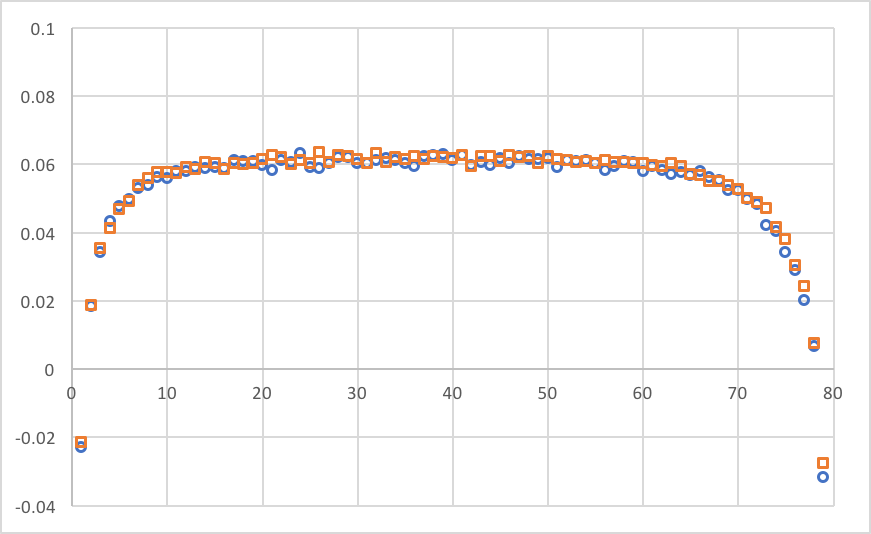} \includegraphics[scale=0.30]{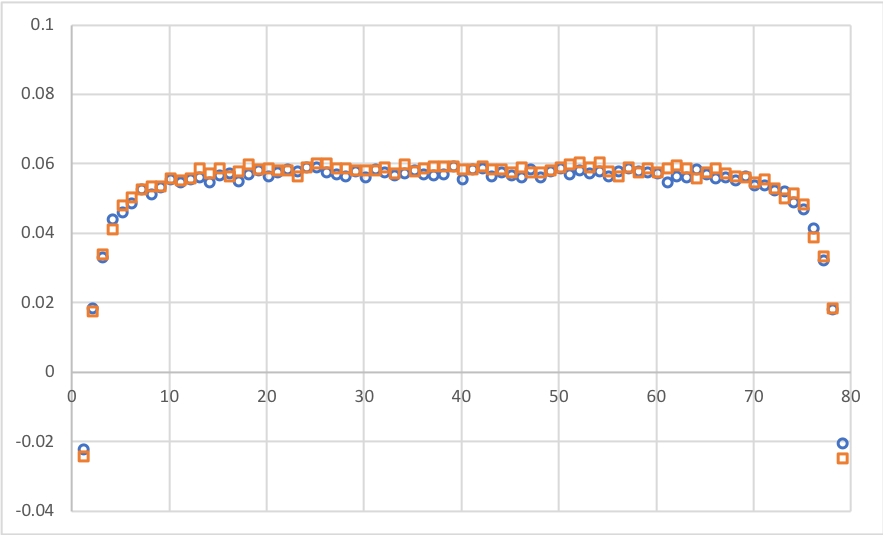} \includegraphics[scale=0.30]{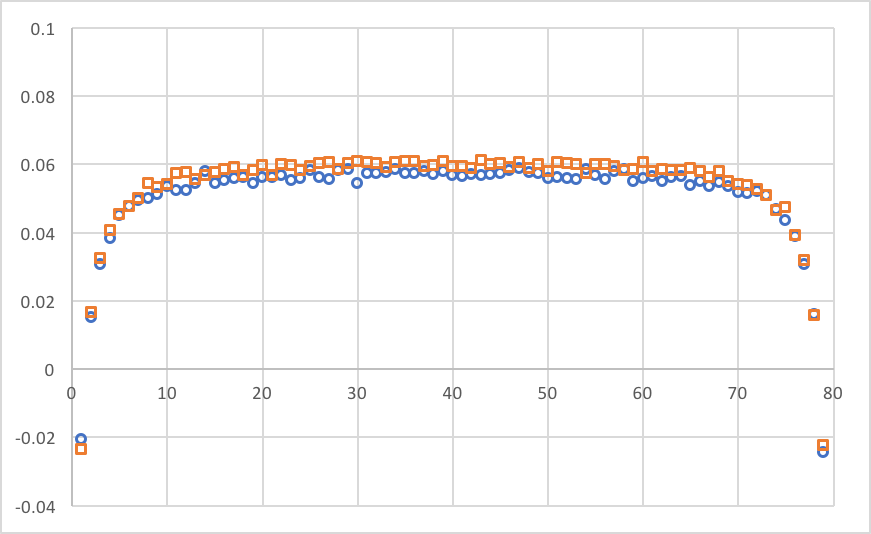}
\includegraphics[scale=0.30]{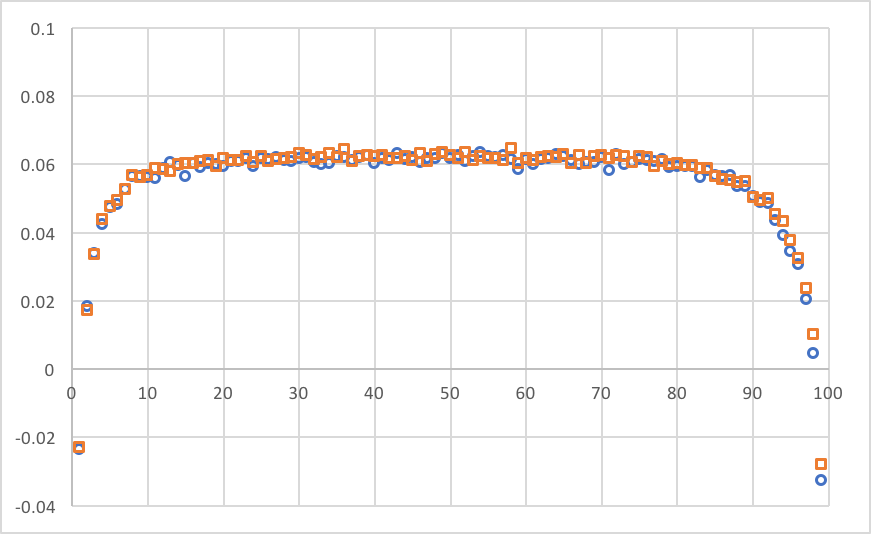} \includegraphics[scale=0.30]{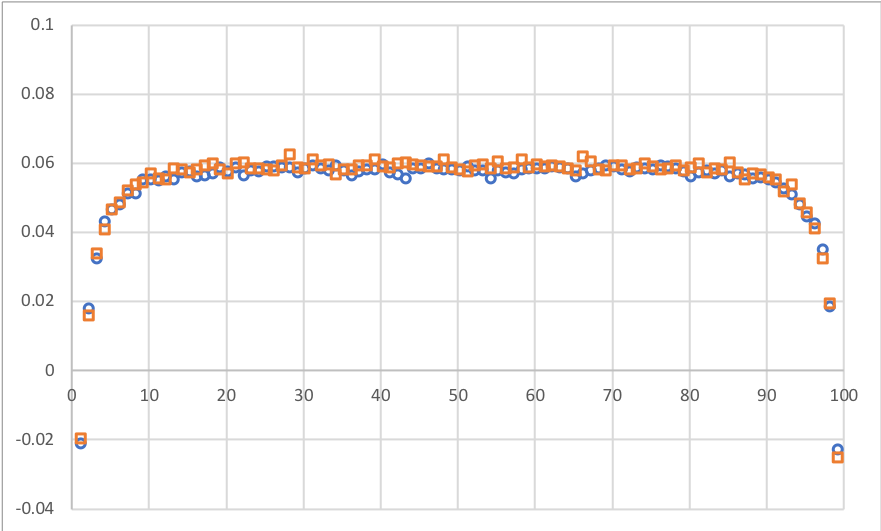} \includegraphics[scale=0.30]{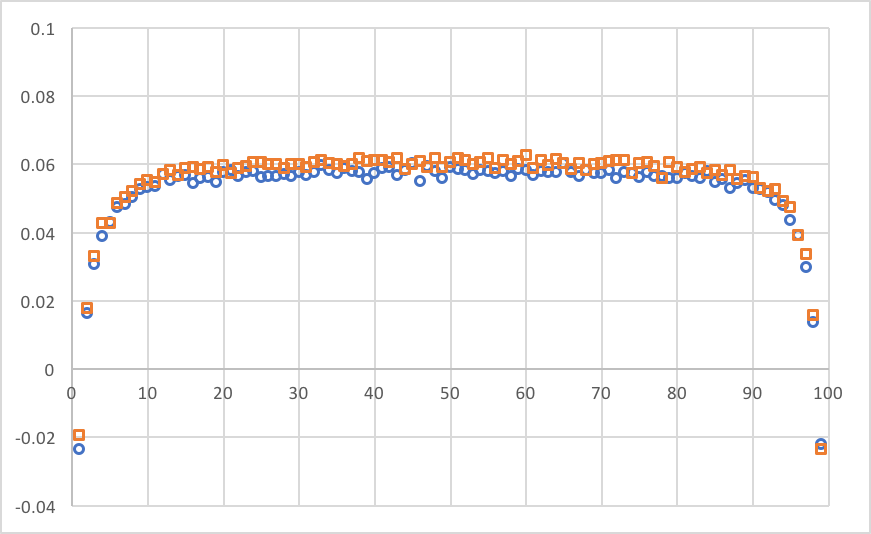}
\includegraphics[scale=0.30]{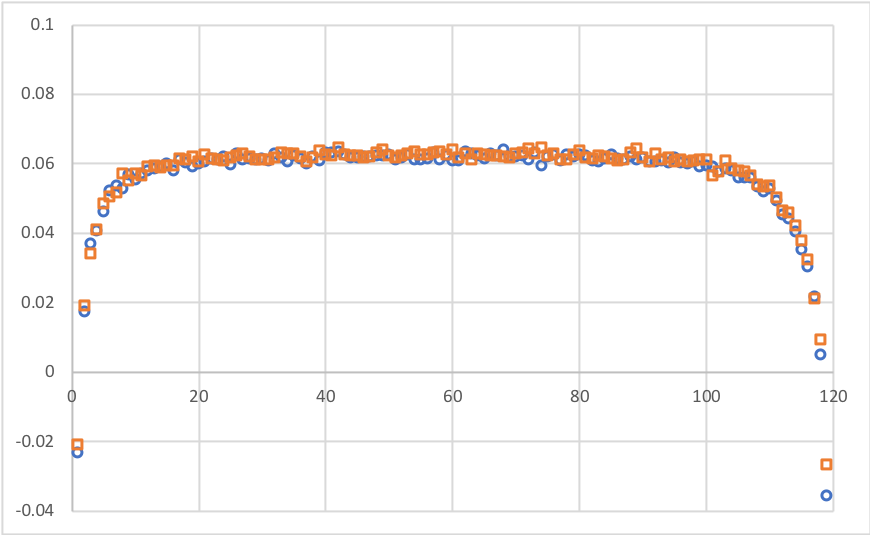} \includegraphics[scale=0.30]{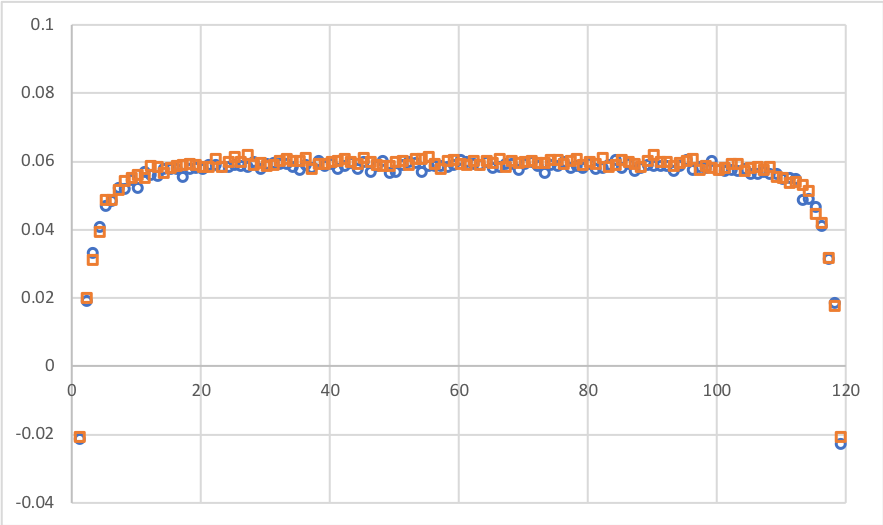} \includegraphics[scale=0.30]{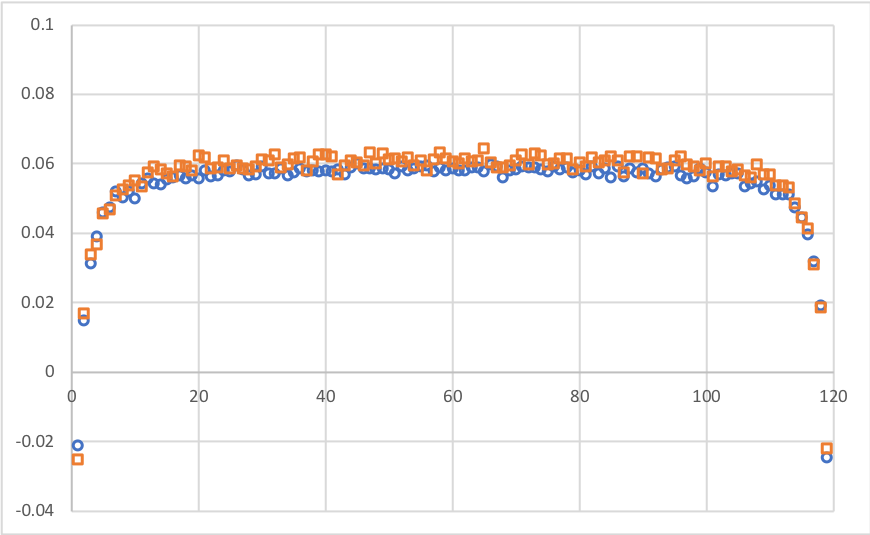}
\caption{Average output of LLL (orange square) and LLL-SP (blue circle). Graphs on each column, from left to right, correspond to the original, random, and greedy choice algorithms, respectively. Graphs on each row represent the results in dimensions 80, 100, and 120, respectively. Within each graph, the horizontal and vertical axes represent the index $k$ on vertices and the average height of the piles $r_k$, respectively.
} \label{fig:output}
\end{figure} 

\begin{table}[]
\begin{tabular}{r|l|l|l|l|l|l|}
\cline{2-7}
\multicolumn{1}{c|}{}     & \multicolumn{2}{c|}{original}                                                                                     & \multicolumn{2}{c|}{random}                                                                                         & \multicolumn{2}{c|}{greedy}                                                                                         \\ \hline
\multicolumn{1}{|c|}{dim} & \multicolumn{1}{c|}{LLL}                                 & \multicolumn{1}{c|}{LLL-SP}                              & \multicolumn{1}{c|}{LLL}                                 & \multicolumn{1}{c|}{LLL-SP}                              & \multicolumn{1}{c|}{LLL}                                 & \multicolumn{1}{c|}{LLL-SP}                              \\ \hline
\multicolumn{1}{|r|}{80}  & \begin{tabular}[c]{@{}l@{}}1.0276\\ 0.00218\end{tabular} & \begin{tabular}[c]{@{}l@{}}1.0273\\ 0.00223\end{tabular} & \begin{tabular}[c]{@{}l@{}}1.0268\\ 0.00206\end{tabular} & \begin{tabular}[c]{@{}l@{}}1.0264\\ 0.00209\end{tabular} & \begin{tabular}[c]{@{}l@{}}1.0267\\ 0.00197\end{tabular} & \begin{tabular}[c]{@{}l@{}}1.0256\\ 0.00197\end{tabular} \\ \hline
\multicolumn{1}{|r|}{100} & \begin{tabular}[c]{@{}l@{}}1.0285\\ 0.00182\end{tabular} & \begin{tabular}[c]{@{}l@{}}1.0282\\ 0.00183\end{tabular} & \begin{tabular}[c]{@{}l@{}}1.0277\\ 0.00172\end{tabular} & \begin{tabular}[c]{@{}l@{}}1.0272\\ 0.00177\end{tabular} & \begin{tabular}[c]{@{}l@{}}1.0276\\ 0.00161\end{tabular} & \begin{tabular}[c]{@{}l@{}}1.0265\\ 0.00167\end{tabular} \\ \hline
\multicolumn{1}{|r|}{120} & \begin{tabular}[c]{@{}l@{}}1.0291\\ 0.00157\end{tabular} & \begin{tabular}[c]{@{}l@{}}1.0288\\ 0.00160\end{tabular} & \begin{tabular}[c]{@{}l@{}}1.0283\\ 0.00151\end{tabular} & \begin{tabular}[c]{@{}l@{}}1.0279\\ 0.00153\end{tabular} & \begin{tabular}[c]{@{}l@{}}1.0282\\ 0.00142\end{tabular} & \begin{tabular}[c]{@{}l@{}}1.0271\\ 0.00142\end{tabular} \\ \hline
\end{tabular}
\caption{Averages and standard deviations of RHF, rounded up to appropriate digits.}
\label{table:rhf}
\end{table}

\begin{figure} 
\centering
\includegraphics[scale=0.30]{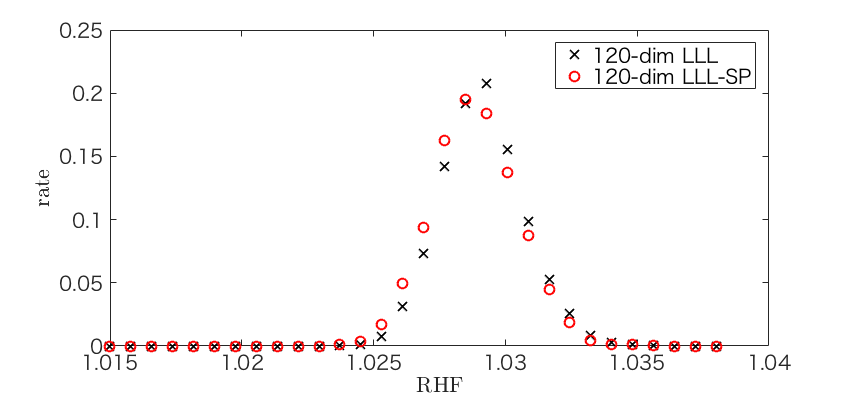}
\includegraphics[scale=0.30]{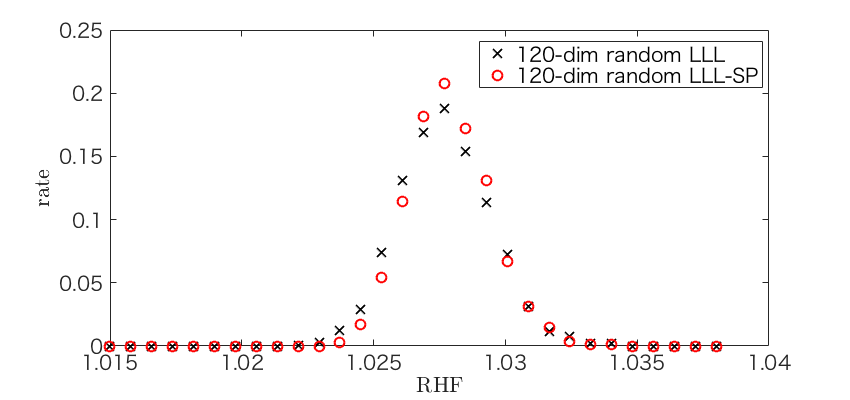} \includegraphics[scale=0.30]{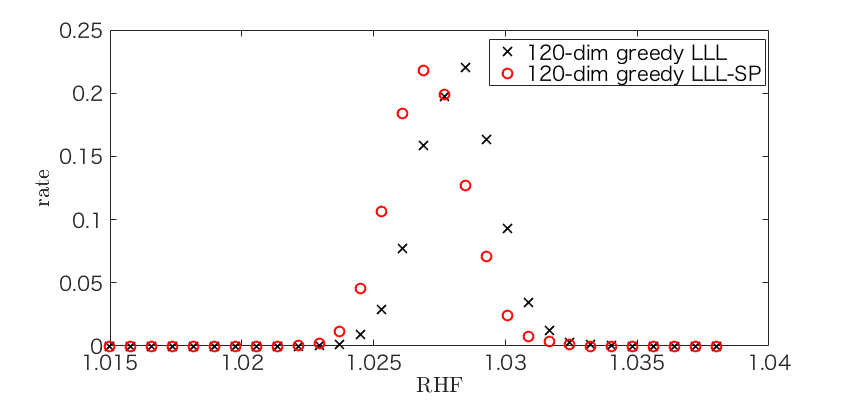}
\caption{Probability distributions of RHFs of LLL and LLL-SP in dimension 120.} \label{fig:RHFdist}
\end{figure} 

The reason that we find LLL and LLL-SP slightly differ with respect to the greedy choice algorithm has to do with the fact that, unlike the original and the random, it ``probes'' one step ahead before making its toppling choice, which has an effect on the $\mu_i$-distribution --- indeed, see Figure \ref{fig:museq} below. We expect this difference to disappear, if LLL-SP is modified to simulate the $\mu_i$-distribution more carefully, using parts of Proposition \ref{prop:post_swap}. Still, it is remarkable that the difference in the average RHF $\approx 0.0011$ is independent of dimension, and the standard deviations remain nearly identical.

The resemblance of the two algorithms runs deeper than on the level of output statistics. See Figures \ref{fig:seq} and \ref{fig:museq}, which depict the plot of points $(i, Q_{k(i)}^{-2})$ and $\mu_{k(i)+1, k(i)} = \mu_{k(i)}$ as we ran LLL and LLL-SP on dimension 80, where $k(i)$ is $k$ chosen at $i$-th iteration. The two plots are again indistinguishable, yet another piece of evidence that LLL and LLL-SP possess nearly identical dynamics. Although too cumbersome to present here, we have the same results on higher dimensions as well.


\begin{figure} 
\centering
\includegraphics[scale=0.47]{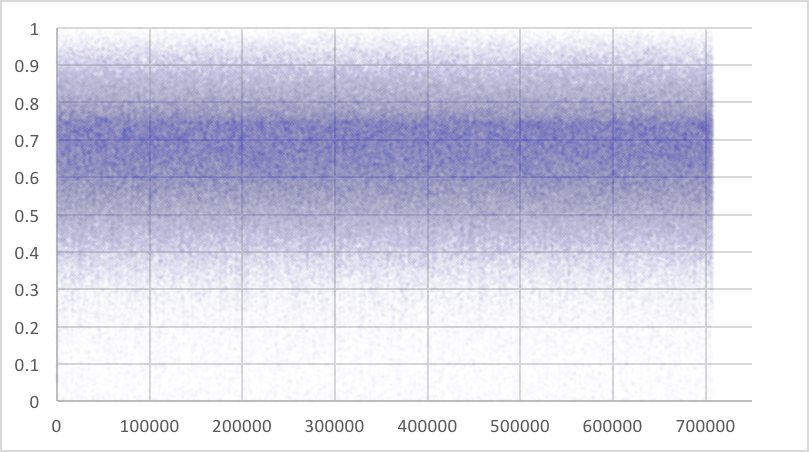} \includegraphics[scale=0.47]{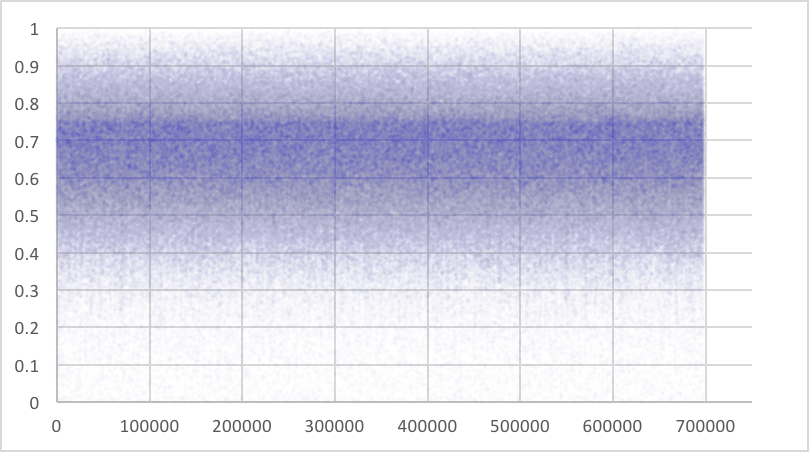}

\includegraphics[scale=0.47]{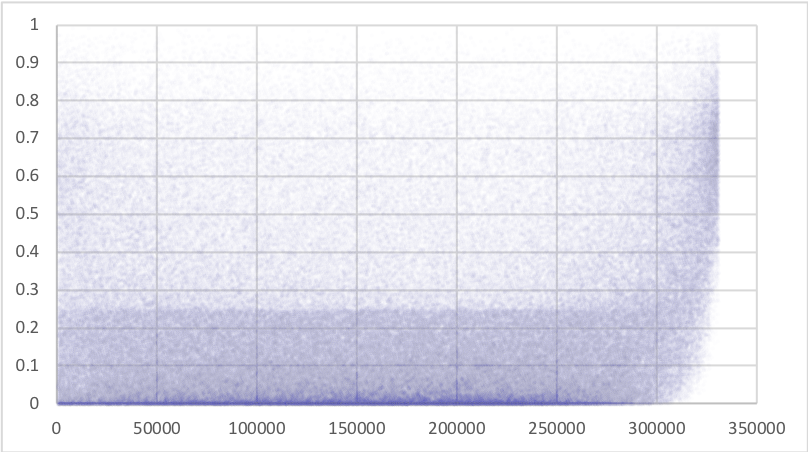} \includegraphics[scale=0.47]{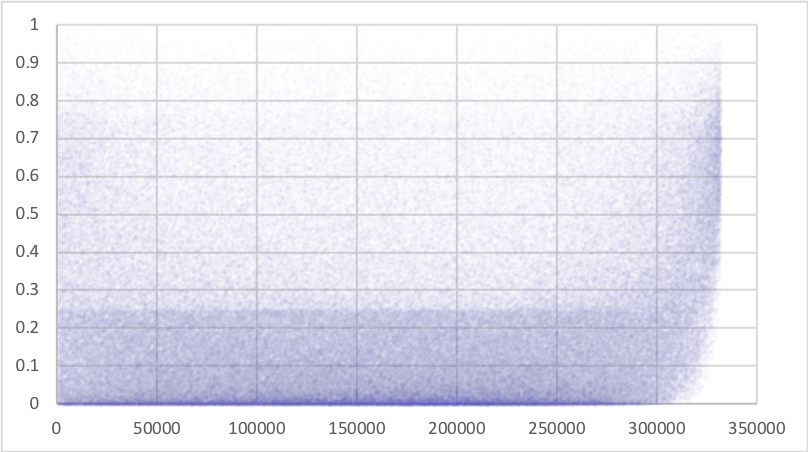}

\includegraphics[scale=0.47]{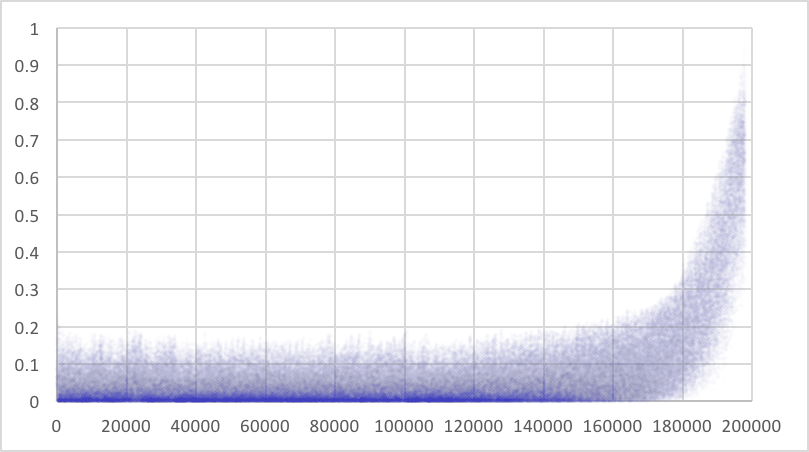} \includegraphics[scale=0.47]{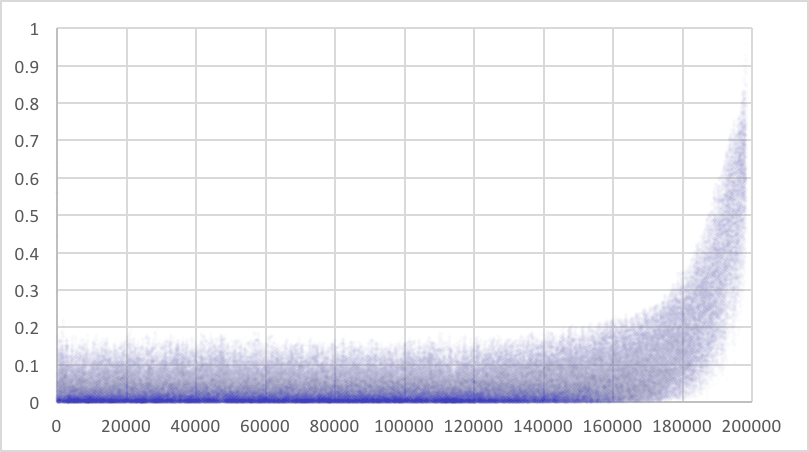}

\caption{Plots of $i$ versus $Q_{k(i)}^{-2}$ during a typical run of LLL(left) and LLL-SP(right), with respect to the sequential, random, and greedy choice algorithms, respectively from top to bottom.} \label{fig:seq}
\end{figure} 

\begin{figure}
\centering
\includegraphics[scale=0.47]{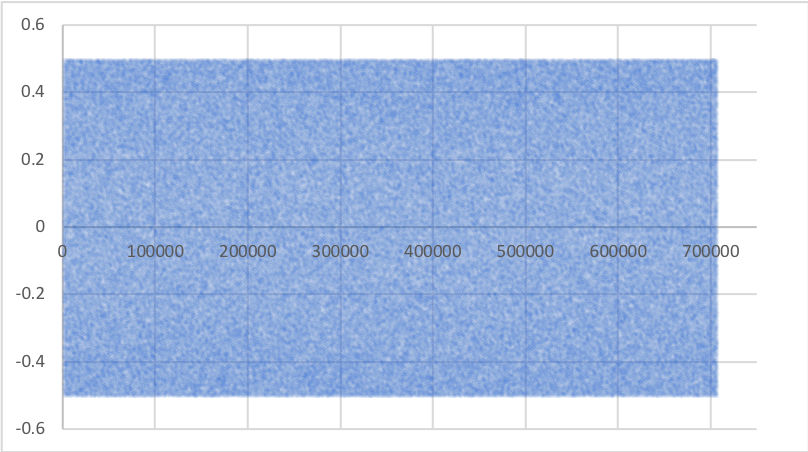} \includegraphics[scale=0.47]{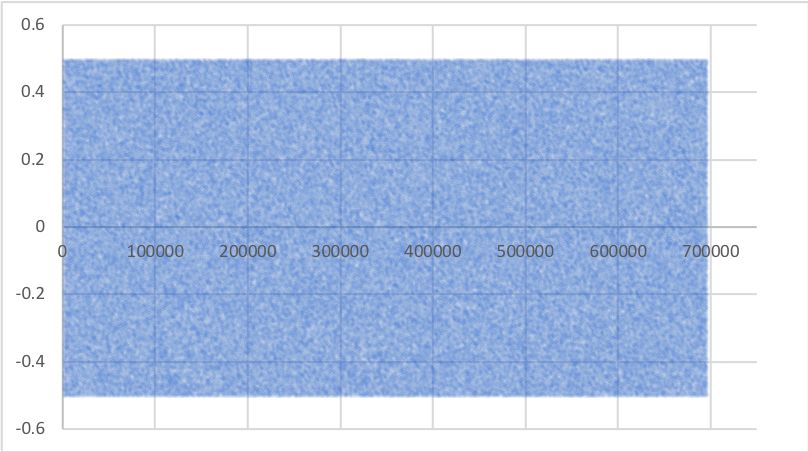}

\includegraphics[scale=0.47]{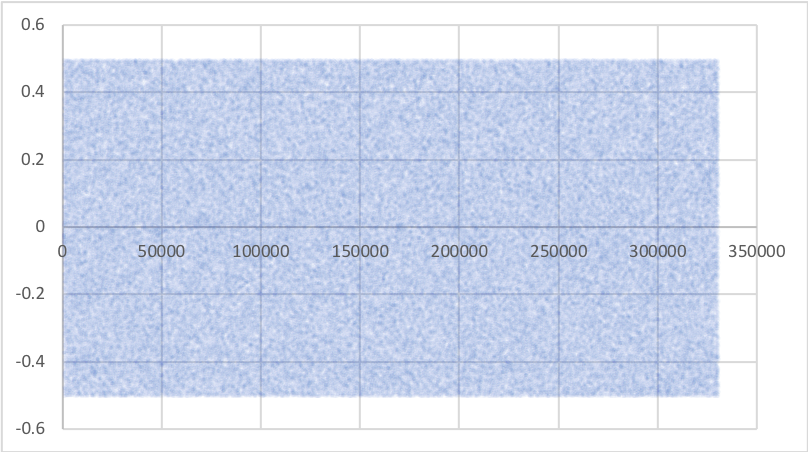} \includegraphics[scale=0.47]{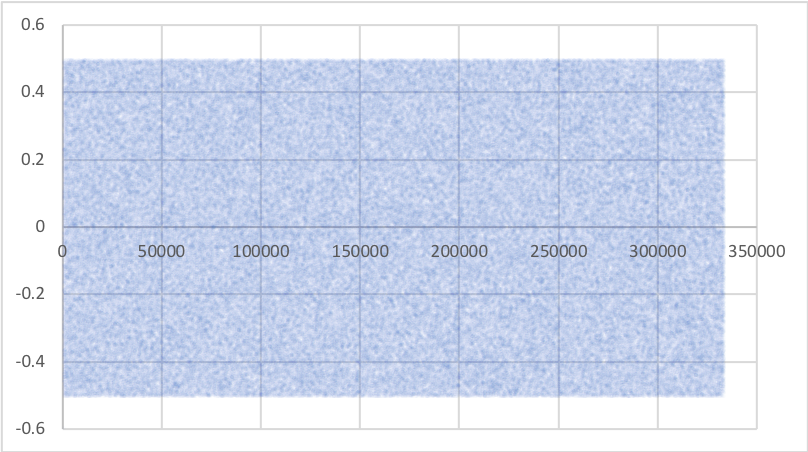}

\includegraphics[scale=0.47]{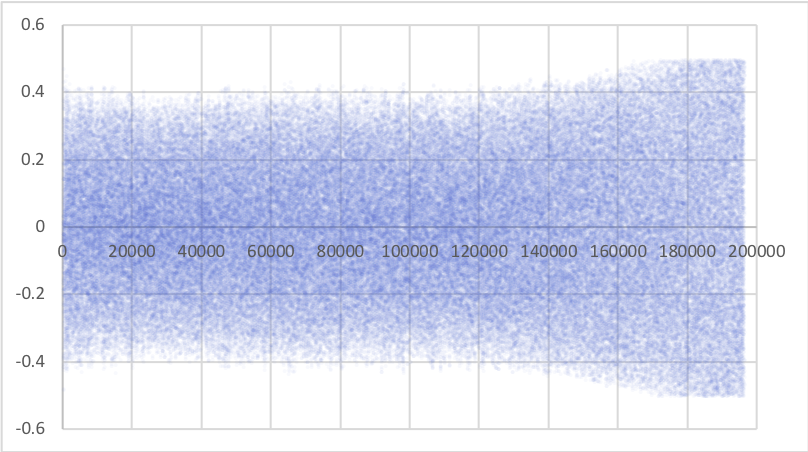} \includegraphics[scale=0.47]{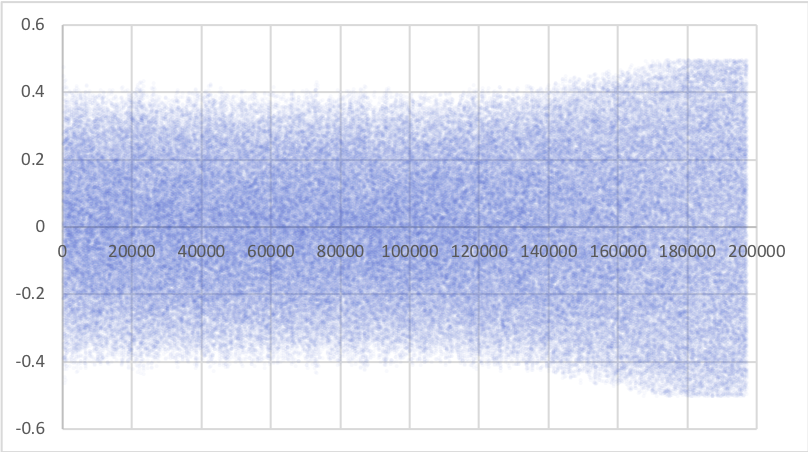}
\caption{Plots of $i$ versus $\mu_{k(i)}$ for LLL(left) and LLL-SP(right), with respect to the sequential, random, and greedy choice algorithms, respectively from top to botton.} \label{fig:museq}
\end{figure}

\subsection{Discussion}

The only difference between LLL and LLL-SP has to do with the way they update the $\mu_k (= \mu_{k+1,k})$'s. For LLL-SP, the $\mu_k$-variables are i.i.d. and independent of the $r_k$-variables. For LLL, $\mu_k$ is determined by a formula involving its previous value and $r_k$. However, it seems plausible that the $\mu_k$'s in LLL, as a stochastic process, is \emph{mixing}, which roughly means that they are close to being i.i.d, in the sense that a small perturbation in $\mu_k$ causes the next value $\mu_k^{new}$ to become near unpredictable. Numerically, this is robustly supported by the graphs at the bottom of Figure \ref{fig:museq}. Theoretically, our intuition comes from the fact that the formula $\mu_{k}^{new} = \mu_{k}/(\mu_k^2 + \alpha_k^{-2})$ (mod 1) is an approximation of the Gauss map $x \mapsto \{1/x\}$, which is well-known to have excellent mixing properties (see e.g. Rokhlin \cite{R61} and the references in Bradley \cite{B05} for more recent works).

The above discussion can be summarized and formulated in the form of a mathematical conjecture, which can then be considered a rigorous version of the statement ``LLL is essentially a sandpile model.'' Below is our provisional formulation of such a conjecture.

\begin{conjecture} \label{conj:mu}
Let $\mathcal{D}$ be a ``generic'' distribution on the set of bases in $\mathbb{R}^n$, to be used to sample inputs for LLL. Define $k(i)$, as earlier, to be the index of the pile toppled at $i$-th iteration, so that $k(i)$ is a random variable depending on the input distribution, and so is $\mu_{k(i)}$. Then

\begin{enumerate}[(i)]
\item The sequence $(|\mu_{k(i)}|)_{i = 1, 2, \ldots}$ is strongly mixing as a stochastic process. (Roughly speaking, this means $|\mu_k(N)|$ is nearly independent of $|\mu_k(M)|$ when $N-M$ is large; see the text \cite{B95} for a precise definition.)
\item Each $|\mu_{k(i)}|$ is contained in a compact subset $S$ of the set of all probability density functions on $[0, 0.5]$ with respect to the $L^\infty$-norm. $S$ is independent of the dimension, the input distribution, or any other variable.
\end{enumerate}
\end{conjecture}

The design intent of Conjecture \ref{conj:mu} is so that what is provable for LLL-SP would also be provable for LLL by an analogous argument (e.g. the theorems in Section 4), while retaining the flexibility as to what the correct distribution of $\mu_k$ might be. It is to be updated accordingly as our understanding of LLL and LLL-SP progresses, in the hope that Conjecture \ref{conj:mu} may come within reach at some point.



\section{Abelian sandpile analogue of LLL, and its RHF gap}

The drawback of LLL-SP as a model of LLL is that, being non-Abelian, it is difficult to study theoretically; indeed, there are few proven results on non-Abelian sandpile models. In this section, we introduce a certain Abelian stochastic sandpile model that we named SSP, which is in a sense an abelianized version of LLL-SP. At a first glance, SSP seems rather removed from LLL, but the shapes of their average output are surprisingly similar. Moreover, SSP admits a mathematical theory that is analogous to that of ASM due to Dhar \cite{D90}, \cite{D06}. This allows us to prove statements such as the average-worst case gap in RHF (Theorem \ref{prop:ssprhf}), suggesting that SSP may be a good starting point for investigating the RHF distributions of reduction algorithms.

We again mention that this section is in fact an exposition of a concurrently written work \cite{Kprep} by SK and YW, slightly rearranged to emphasize the connection to LLL. Although we transferred much of our work on SSP to a separate paper in order to properly treat it from the physical perspective, we offer its detailed summary for the completeness of our narrative here.

\subsection{Background on ASM}

To facilitate the reader's understanding, we briefly describe the Abelian sandpile model (ASM), the most basic of sandpile models, and parts of its theory that is relevant to us. Its pseudocode is provided in Algorithm \ref{alg:asm}. See Dhar \cite{D90}, where the theory is originally developed, or the presentation slides by Perkinson \cite{P14}.

\begin{algorithm}
\caption{Abelian sandpile model (ASM)}\label{alg:asm}
\begin{enumerate}[1.]
\item[0.] Input: $r_1, \ldots, r_{n-1} \in \mathbb{Z}$, parameters $T, I \in \mathbb{Z}$, $0 < I \leq T/2$
\item while true, do:
\item \hspace{4mm} choose a $k \in \{1, \ldots, n-1\}$ such that $r_k > T$
\item \hspace{4mm} if there is no such k, break
\item \hspace{4mm} subtract $2I$ from $r_k$
\item \hspace{4mm} if $k > 1$, add $I$ to $r_{k-1}$; if $k < n$, add $I$ to $r_{k+1}$
\item Output: integers $r_1, \ldots, r_{n-1} \leq T$
\end{enumerate}
\end{algorithm}

The important ASM concepts for us are that of the \emph{recurrent configurations} and the \emph{steady state}. Let $M$ be the set of all stable (non-negative) configurations of ASM. Given two configurations $r, s \in M$, we have the operation
\begin{equation*}
r \oplus s = \mbox{(stabilization of $r+s$)},
\end{equation*}
which is the outcome of ASM with input being the configuration $r+s$ defined by $(r+s)_i = r_i + s_i$ for each $i$. Unlike LLL, the output of ASM is independent of the choice of toppling order --- hence the term ``Abelian'' --- and thus $\oplus$ is well-defined. This operation makes $M$ into a commutative monoid.

Define $g \in M$ to be the configuration with $g_1 = 1$ and $g_2 = \ldots = g_{n-1} = 0$. We call $r \in M$ \emph{recurrent} if
\begin{equation*}
\underbrace{g \oplus \ldots \oplus g}_{m\, \mathrm{times}} = \mbox{$r$ for infinitely many $m$}.
\end{equation*}
One can actually take any $g$ for which at least one $g_i$ is coprime to the g.c.d. of $T$ and $I$ (this condition is only to avoid concentration on a select few congruence classes). Equivalently, with LLL in mind, we can also define that $r$ is recurrent if there exist infinitely many non-negative input configurations such that their stabilization results in $r$. It is a theorem that the set $R$ of the recurrent configurations of ASM forms a group under $\oplus$.

(Note for the experts: these definitions of recurrent configurations may be rather unconventional, but are equivalent to the standard formulations, e.g. the one in \cite{D90}. It is a simple exercise to prove the equivalence, under the following setting: interpret the space of all configurations as $\Z^{n-1}$ and consider the orbits of the toppling operators, which are cosets of a certain sublattice of $\Z^{n-1}$. In each coset, there exists exactly one configuration to which infinitely many non-negative configurations on the same coset stabilizes.)

One may ask, given an $r \in R$, what is the proportion of $m \in \Z_{>0}$ that satisfies $g \oplus \ldots \oplus g\, (m\, \mathrm{times}) = r$? It turns out that the answer is $1/|R|$ for any $r \in R$, that is, each element of $R$ has the same chance of appearing. This distribution, say $\rho$, on $R$ is called the \emph{steady state} of the system. And the phrase \emph{average output shape} that we have been using in the empirical sense obtains a formal definition as $\sum_{r \in R} \rho(r)r$. The steady state is unique in the following sense: choose an $r \in R$ according to $\rho$, and take any configuration $s$; then $r \oplus s$ is again distributed as $\rho$.


\subsection{Introduction to SSP}

A pseudocode for SSP is provided in Algorithm \ref{alg:ssp}. This is exactly the same as ASM, except for Step 4, which determines the amount of sand to be toppled at random. The decision to sample from the uniform distribution is an arbitrary one; we could have chosen any compactly supported distribution, and much of the discussion below still applies.

\begin{algorithm}
\caption{Stochastic sandpile (SSP)}\label{alg:ssp}
\begin{enumerate}[1.]
\item[0.] Input: $r_1, \ldots, r_{n-1} \in \mathbb{Z}$, parameters $T, I \in \mathbb{Z}$, $0 < I \leq T/2$
\item while true, do:
\item \hspace{4mm} choose a $k \in \{1, \ldots, n-1\}$ such that $r_k > T$
\item \hspace{4mm} if there is no such k, break
\item \hspace{4mm} sample $\gamma$ uniformly from $\{1, \ldots, I\}$
\item \hspace{4mm} subtract $2\gamma$ from $r_k$
\item \hspace{4mm} augment $\gamma$ to $r_{k-1}$ and $r_{k+1}$
\item Output: integers $r_1, \ldots, r_{n-1} \leq T$
\end{enumerate}
\end{algorithm}

The average output shape of this stochastic sandpile model (SSP) is shown in Figure \ref{fig:sspoutput}. Figure \ref{fig:sspoutput} shares all the major characteristics of Figure \ref{fig:output}: flat in the middle, and diminishing at both ends. In cryptographic literature these features have been respectively referred to as the geometric series assumption(GSA) and its failure at the boundaries. In Section 5, we will see that finite-size scaling theory provides a far more quantitatively robust description of the output shape.

\begin{figure}
\centering
\includegraphics[scale=0.6]{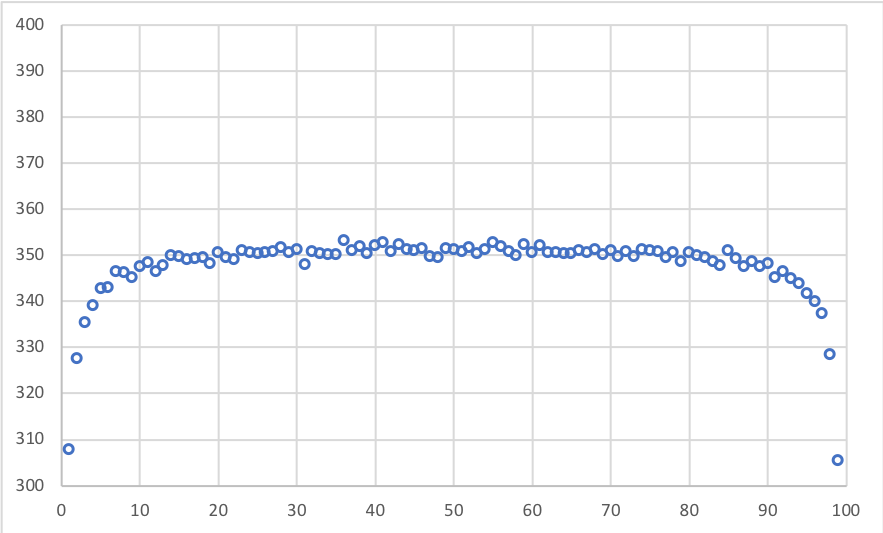}
\caption{Average output of SSP, $n = 100$, $I = 200$ and $T = 400$.} \label{fig:sspoutput}
\end{figure}

\subsection{Mathematical properties of SSP}

A mathematical theory of SSP closely analogous to that of ASM has been recently developed in \cite{Kprep}, largely motivated by the experimental result of the previous section. Every aspect of the above-mentioned ASM theory carries over to the SSP theory, except that instead of configurations one works with a distribution on the set of configurations, due to its stochastic nature. For configurations $r^{(1)}, \ldots, r^{(k)}$ and $p_i \in (0,1]$ such that $p_1 + \ldots + p_k = 1$, we write
\begin{equation} \label{eq:sspelt}
\sum_{i=1}^k p_i[r^{(i)}]
\end{equation}
to represent a distribution that assigns probability $p_i$ to the configuration $r^{(1)}$. For instance, if $r$ is a configuration unstable at vertex $i$, and if $v_i = (0, \ldots, -1, 2, -1, \ldots, 0)$ with $2$ in $i$-th entry, then for the toppling operator $T_i$ we have
\begin{equation} \label{eq:ssptopple}
T_i[r] = \sum_{\gamma=1}^I \frac{1}{I}[r - \gamma v_i].
\end{equation}


We say a configuration of form \eqref{eq:sspelt} is \emph{mixed} if $k \geq 2$ and \emph{pure} otherwise, \emph{stable} if all $r^{(i)}$'s are stable, and \emph{nonnegative} if all $r^{(i)}$'s are nonnegative.

The most important property of SSP is that, like ASM, it possesses a unique steady state, that is, a mixed configuration $g$ such that
\begin{equation*}
g \oplus f = g
\end{equation*}
for any nonnegative $f$. It is clear that if we understand the steady state, then we understand the RHF distribution. The following is easy to prove:

\begin{theorem}\label{prop:ssprhf}
SSP possesses a unique steady state. The worst-case $\log\mathrm{(RHF)}$ of SSP is $T/2 + o_n(1)$. The average $\log\mathrm{(RHF)}$ of SSP is bounded from above by $T/2 - I/2e^2 + o_n(1)$.
\end{theorem}

We note that empirically one observes $\log\mathrm{(RHF)} \approx T/2 - I/8$ on average.

\begin{proof}[Sketch (and discussion) of proof]

This is essentially Proposition 8 of \cite{Kprep}. We present the sketch of the proof for completeness. Most of the argument is devoted to the existence of the steady state, from which the rest of the theorem follows.

\begin{figure}
\centering
\includegraphics[scale=0.3]{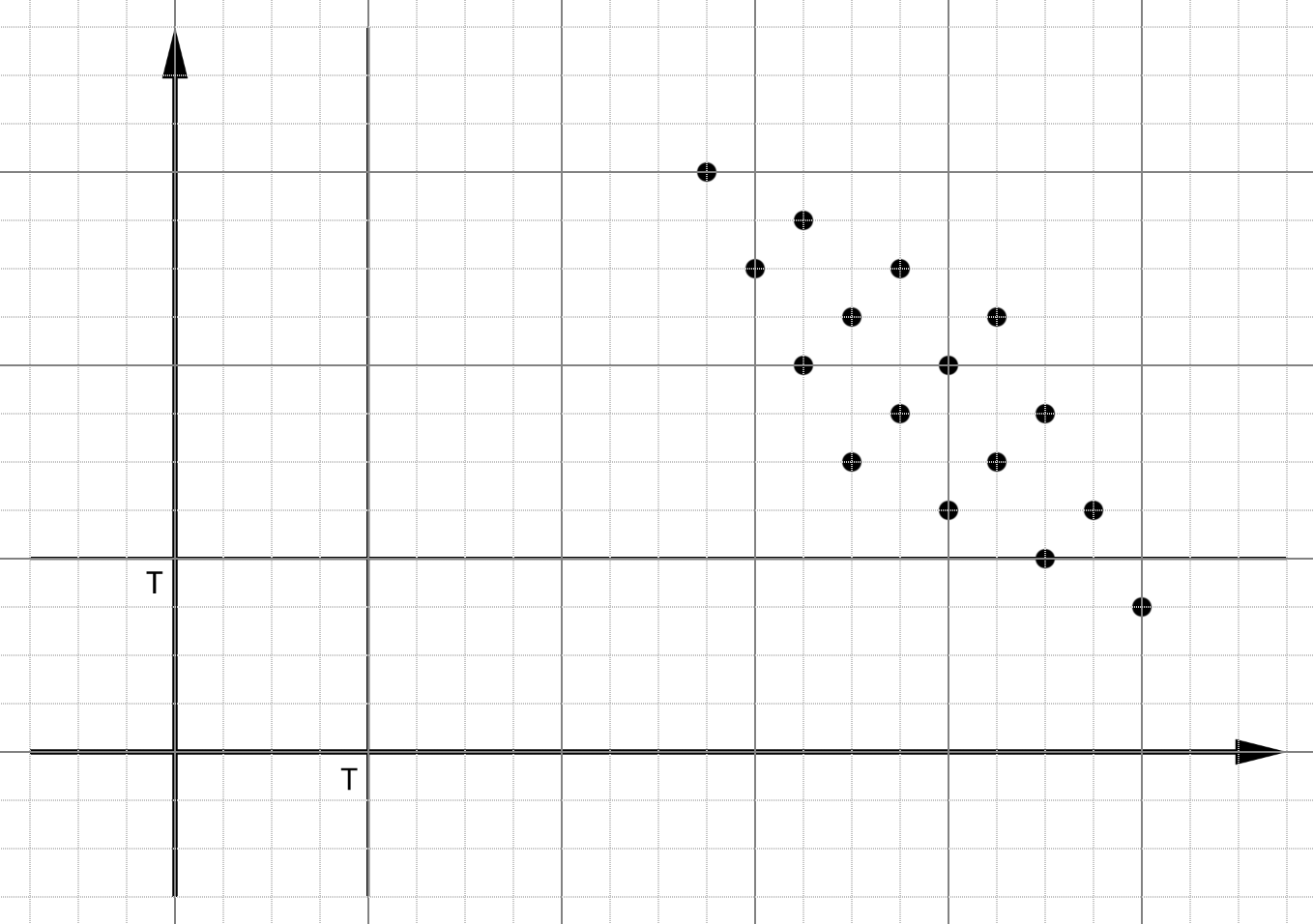} \vspace{3mm}

\includegraphics[scale=0.3]{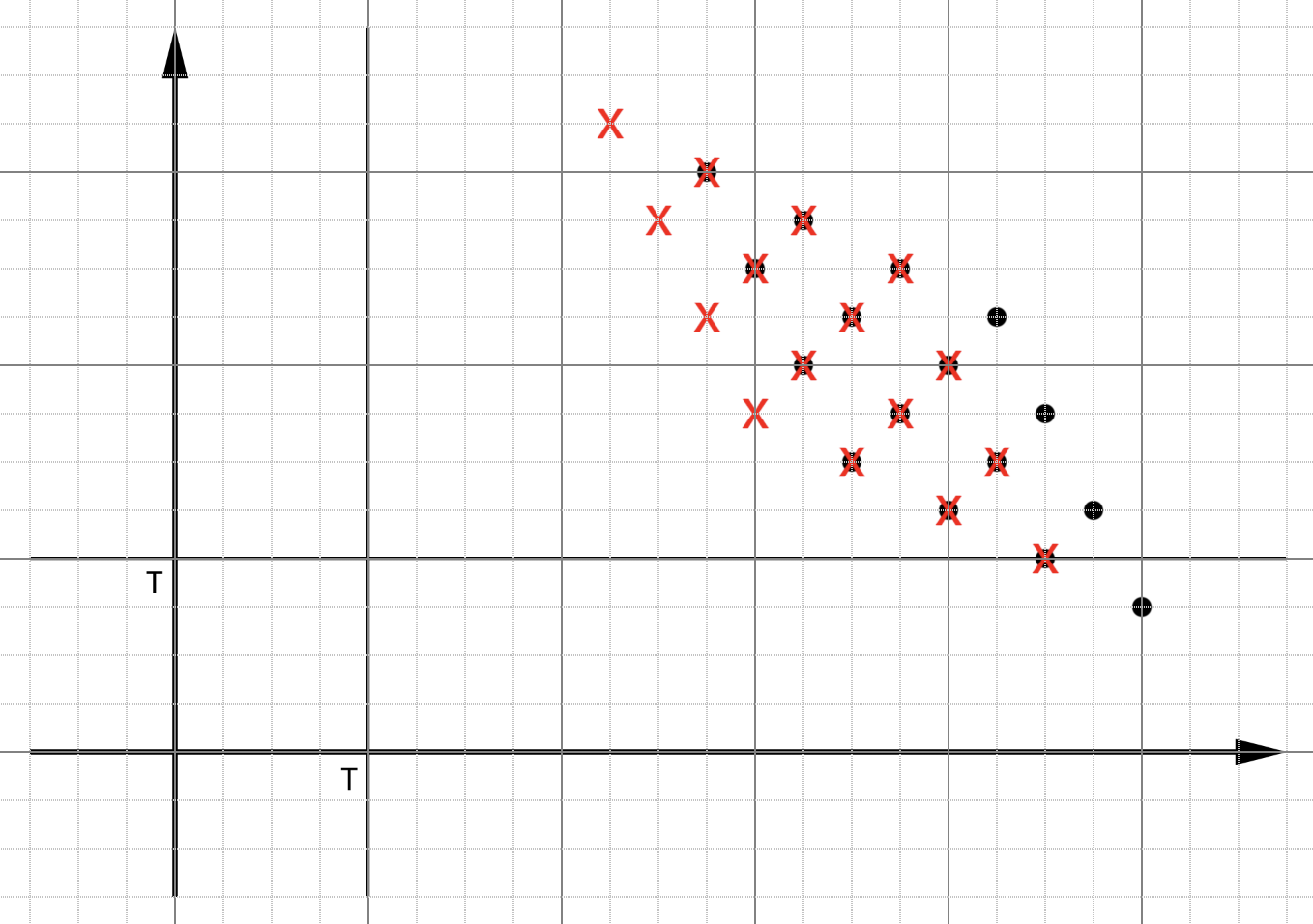} \vspace{3mm}

\includegraphics[scale=0.3]{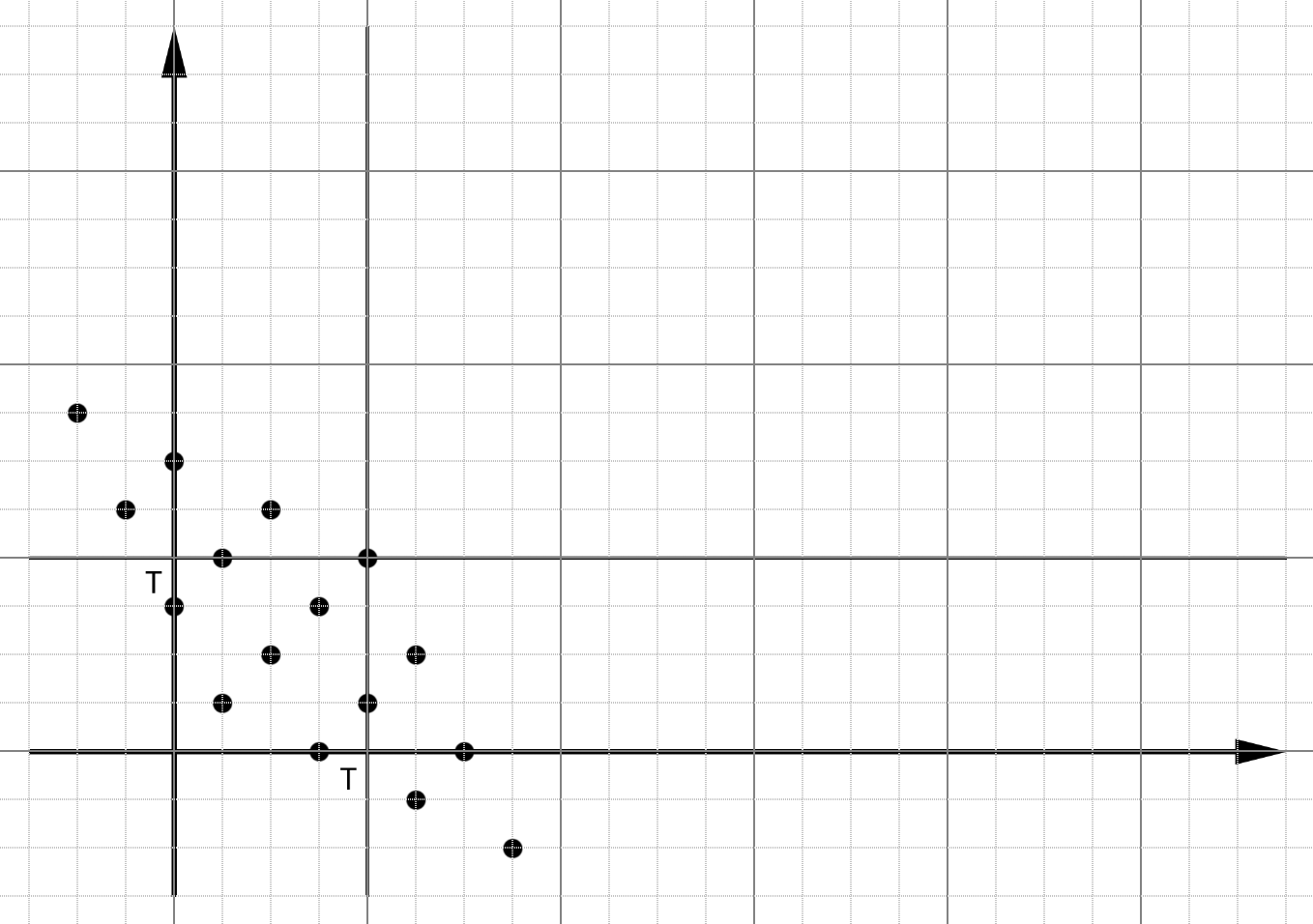} \vspace{2mm}

\caption{The parallelepiped argument.} \label{fig:par}
\end{figure} 

Take an unstable (pure) configuration $r$. If $r$ is sufficiently far away from the origin in the configuration space, we must topple on each and every vertex at least once --- in fact, arbitrarily many times --- on the way of stabilizing $r$. So consider $T_1T_2 \ldots T_{n-1}[r]$, where $T_i$ is the toppling operator on vertex $i$. By repeated applications of \eqref{eq:ssptopple}, $T_1T_2 \ldots T_{n-1}[r]$ is a distribution on the configuration space that is supported on a parallelepiped-shaped cluster, as illustrated in the top of Figure \ref{fig:par} in case $n = 3$ and $I = 4$; the upper-right vertex in the parallelogram is $r - (1, 1, \ldots, 1)$.

Applying $T_i$ to this parallelepiped-shaped distribution amounts to ``pushing'' the parallelepiped in the direction of $i$, resulting in another parallelepiped-shaped distribution. The middle graph in Figure \ref{fig:par} illustrates this process, by indicating with x marks the outcome of applying $T_1$ to the original distribution (assuming that the horizontal axis represents $r_1$). Repeating, we eventually reach the situation as in the bottom of Figure \ref{fig:par}, where none of the $T_i$ would preserve the shape of the parallelepiped, since $(T, T, \ldots, T)$ is already a stable configuration and thus $T_i$ leaves it there. From this point on, the action of $T_i$ can no longer be easily described.

However, we claim that, for any $r$ sufficiently far enough from the origin, the distribution on the parallelepiped obtained by the time the upper-right corner reaches $(T, \ldots, T)$ is arbitrarily close to a certain limiting distribution $\wp$. To see this, consider the action of $T_i$ on the distribution on the parallelepiped, while forgetting the information about where that parallelepiped is located in the configuration space. Then one notices that each $T_i$ acts as a linear operator on the space of such distributions. Simultaneously diagonalizing all $T_i$'s --- possible because they pairwise commute --- one finds that $1$ is the single largest eigenvalue of multiplicity one, whose corresponding eigenvector is $\wp$. Upon repeated applications of $T_i$'s, the components corresponding to the lesser eigenvalues converge to zero, proving the claim. This proves that SSP has a unique steady state.

In fact, $\wp$ can be easily computed, allowing us to show that the maximum point density of the steady state occurs at $(T, \ldots, T)$ with density $\approx (I/2)^{-(n-1)}$. This is enough to deduce a nontrivial upper bound on the average RHF, as follows. Estimate the number $N(\alpha)$ of stable configurations whose $\log\mbox{(RHF)}$ are greater than $\alpha$, and take $\alpha$ such that $N(\alpha) \cdot (I/2)^{-(n-1)}$ vanishes as $n \rightarrow \infty$. It turns out we can choose $\alpha = T/2 - I/2e^2$.

\end{proof}

There are a couple of difficulties in directly applying the same idea to LLL or LLL-SP. For instance, because the increment depends on the $r_i$'s for those systems, the effect of $T_i$ is not as neat as illustrated in Figure \ref{fig:par}. It would push the side of the parallelepiped with ``uneven force,'' skewing the shape of the parallelepiped and the distribution lying on it.
This makes proving the existence of the steady state for LLL or LLL-SP difficult.

However, for the purpose of bounding the average RHF away from the worst-case, all we need to show is that the maximum density of the output distribution cannot be too large. This seems feasible yet quite vexing; we state it as a conjecture below for future reference. As in the SSP case, we expect that the maximum density is attained on the upper-right corner.

\begin{conjecture} \label{conj:mass}
For a generic distribution $\mathcal{D}$ on the set of bases of $\R^n$, the probability density function of the corresponding output distribution $\mathcal{D}^\circ$ of LLL (or LLL-SP) is bounded from above by a constant $C$ that depends only on $n$.
\end{conjecture}

It may also be interesting to try to deduce other statements on the RHF of SSP, e.g. a lower bound on the average RHF, or why the average RHF appears to be Gaussian, as in Figure \ref{fig:RHFdist}.

\section{Regarding time complexity}

Although expanding the SSP theory, and Theorem \ref{prop:ssprhf} in particular, to LLL-SP seems challenging for the time being, we are able to prove some attractive statements for LLL-SP with respect to its complexity, which we present below. We also consider their extensions to LLL assuming the truth of Conjecture \ref{conj:mu}.

\subsection{A lower bound}

The theorem below gives a probabilistic \emph{lower} bound on the complexity of LLL-SP, which agrees up to constant factor with the well-known upper bound. There are two ingredients in the proof: (i) measuring the progress of the LLL algorithm by the quantity \emph{energy}, a well-known idea from the original LLL paper \cite{LLL82} (ii) bounding the performance of LLL-SP by a related SSP.

\begin{theorem} \label{thm:lowertime}
Consider LLL-SP, and an input configuration $r$ whose \emph{log-energy} $E = E(r)$, defined by
\begin{equation*}
E(r) = \sum_{j=1}^{n-1} \sum_{i=j}^{n-1} (n-i)r_i,
\end{equation*}
is sufficiently large --- in fact, $E > 10H$ works, with $H$ defined as in \eqref{eq:enbound}. Then the probability that LLL-SP is not terminated in $E/4$ steps is at least $1 - CE^{-1/2}$ for an absolute constant $C > 0$.
\end{theorem}

Observe that the familiar upper bound $O(n^2\log \max_i\|\b_i\|)$ on the number of required steps is equivalent to $O(E)$, with the implicit constant depending on $\delta$.

\begin{proof}
If the algorithm is terminated, then $E$ must have become less than
\begin{equation*}
\sum_{i=1}^{n} (n-i+1)(n-i)T/2,
\end{equation*}
(where $T := -\log \delta^{1/2} > 0$) which equals,
\begin{equation} \label{eq:enbound}
H:= \frac{T}{6}(n^3 - n).
\end{equation}
Taking the converse, we see that if $E$ is greater than \eqref{eq:enbound}, then LLL-SP has not yet terminated. At $i$-th toppling, $E$ decreases by at most $\log \mu_{k(i)}^{-2}$, where $k(i)$ is the index of the vertex in which $i$-th toppling occured. If toppled $N$ times, the decrease in $E$ is bounded by at most $F_N := \sum_{i=1}^N \log \mu_{k(i)}^{-2}$. In sum,
\begin{equation} \label{eq:lowertimegoal}
\mathrm{Prob}(E - F_N > H)
\end{equation}
gives the lower bound on the probability that LLL-SP is not terminated after $N$ swaps. Hence, it suffices to show that \eqref{eq:lowertimegoal} is bounded from below by $1-CE^{-1/2}$ when $N = E/2$.

The central limit theorem is applicable on $F_N$, since $\mu_{k(i)}$ are i.i.d.
More precisely, we apply the Berry-Esseen theorem, which asserts the following. Suppose we have i.i.d. random variables $X_1, X_2, \ldots$, so that $m = \mathbb{E}(X_1)$, $\sigma = (\mathbb{E}(X_1^2) - \mathbb{E}(X_1)^2)^{1/2}$, and $\rho = \mathbb{E}(X_1^3)$ are all finite. Furthermore, let $Y_N = \sum_{i=1}^N X_i$, and let $G_N(x)$ be the cumulative distribution function of $Y_N$, and $\Phi_N(x)$ be the cumulative distribution function of the normal distribution $N(Nm, N\sigma^2)$. Then for all $x$ and $N$,
\begin{equation*}
\left| G_N(x) - \Phi_N(x) \right| = O(N^{-1/2}),
\end{equation*}
where the implied constant depends on $m, \sigma, \rho$ only.

We let $X_i = \log \mu_{k(i)}^{-2}$ so that $F_N = G_N$, and apply the Berry-Esseen theorem. It is easy to compute and check that $m, \sigma, \rho$ are all finite e.g. $m = 2(1+\log 2) \approx 3.386$ and $\sigma = 2$. Then, for a random variable $\mathcal{N}_N \sim \mathcal N(Nm, N\sigma^2)$, \eqref{eq:lowertimegoal} is bounded by
\begin{equation*}
\mathrm{Prob}(E - \mathcal{N}_N > H)
\end{equation*}
plus an error of $O(N^{-1/2})$.

Now choose $N = E/4$, so that $\mathcal{N}_N \sim \mathcal N((1+\log 2)E/2, E)$. Using Chebyshev's inequality we can prove
\begin{equation*}
\mathrm{Prob}(\mathcal{N}_N \geq 0.9E) \leq O(E^{-1}),
\end{equation*}
where the implied constant is absolute. Thus if $E$ is large enough so that $E - H > 0.9E$, we have that \eqref{eq:lowertimegoal} is at least $1 - CE^{-1/2}$ for some $C > 0$, as desired.

\end{proof}

\begin{remark}
1. We can use the same idea to obtain a lower bound on the average RHF of LLL-SP, but it turns out to be slightly less than $1$, which happens to be useless in the context we are in.

2. There exists a central limit theorem for a strong mixing process (see \cite{B95}), and also a central limit theorem for a sequence of independent but non-identical sequence of random variables (e.g. the Lyapunov CLT). Conjecture \ref{conj:mu} states that the $|\mu_{k(i)}|$ of LLL is strong mixing (weaker than independent) and non-identical (though contained in a compact set). We do not know whether there exists a central limit theorem that applies in this context, though we suspect that there should be.

\end{remark}

\subsection{The optimal LLL problem}

The optimal LLL problem (see e.g. \cite{A00}) asks whether LLL with the optimal parameter $\delta = 3/4$ terminates in polynomial time. The following theorem, while crude, shows that this is true for LLL-SP with arbitrarily high probability.

\begin{theorem} \label{thm:optimal}
For any $\eta > 0$ small, LLL-SP with $\delta = 3/4$ terminates after $O_\eta(E)$ steps with probability $1 - \eta$.
\end{theorem}
\begin{proof}

Write $\mu$ for the random variable uniformly distributed in $[0, 1/2]$. In the case $\delta < 3/4$, the complexity bound of LLL is established with the observation that, with each swap, the energy $E$ decreases by at least $c := \log(\delta + 1/4)^{-1} > 0$, and thus the algorithm must terminate within $E/c$ steps. Similarly, in the case $\delta = 3/4$, we try to show that the minimum change of energy $\log(\delta + \mu^2)^{-1}$ is strictly bounded away from zero almost all the time.

(If $I$ was the increment for a given toppling operation, it is easy to show that the energy decreases by $2I$ after such a step.)

Choose a small $\varepsilon > 0$, and let $p = \mathrm{Prob}(\mu \leq (1-\varepsilon)/2) = 1- \varepsilon$. Let $d = \log(3/4 + p^2/4)^{-1}$, which is the minimum possible change in energy provided $\mu \leq (1-\varepsilon)/2$. Now take $10E/d$ samples $\mu_1, \mu_2, \ldots$ of $\mu$ (there is nothing special about the constant $10$ here). If at least $E/d$ of those samples are less than $(1-\varepsilon)/2$, LLL-SP would terminate. Proving that this probability is arbitrarily close to $1$ is now a simple exercise with the binomial distribution.

\end{proof}

Observe that the above proof carries over to the case of LLL assuming Conjecture \ref{conj:mu}; the compactness condition on the $\mu_{k(i)}$ distributions allows control on the probability that they are all simultaneously bounded away from $1/2(1-\varepsilon)$.

\section{Finite-size scaling theory}

Finite-size scaling (FSS) is a theory in statistical physics used to study critical phenomena. Such phenomena are often studied via models on finite graphs and then by analyzing the quantity $\chi$ of interest as the system size $L$ --- the number of vertices of the graph --- goes to infinity. Roughly speaking, FSS asserts that, upon a proper rescaling of the variables, $\chi$ becomes nearly independent of $L$ for $L \gg 0$. FSS also provides a description of this asymptotic behavior of $\chi$ as $L \rightarrow \infty$.

For sandpile models, FSS implies asymptotic formulas that would be particularly interesting if they also applied to the LLL algorithm, as discussed in Section 1.2 above. Although it would be inappropriate to say ``apply FSS to LLL,'' as LLL has no underlying critical phenomenon, the formulas themselves, isolated from the context of the original theory, could certainly be tried. We ran a long experiment on LLL that is analogous to the one in Section III of Grassberger, Dhar, and Mohanty \cite{GDM16}, in which the authors employ FSS to study the Oslo model, a sandpile model with entirely different toppling rule than the ones we have considered so far. This section presents the results from this experiment.

\subsection{A brief introduction to FSS}

We start with a brief introduction to FSS and its predictions that are pertinent to our work. For readers who are unfamiliar with physics but wish to gain some quick basic knowledge, we recommend browsing the theory of one- and two-dimensional Ising models. Also see Section III of \cite{GDM16}, which states the formulas \eqref{eq:1st}-\eqref{eq:3rd} that we will introduce below. For more serious general treatises on FSS, see \cite{C88} or \cite{G18}.

In the theory of critical phase transition in physics --- e.g. the transition in a magnetic material from a magnetized to unmagnetized state --- one finds that the quantity $\chi$ of interest, for example the magnetic susceptibility, diverges near the critical point, or critical temperature; see Figure \ref{fig:fss1}. Furthermore, this divergence is often described by a power law, e.g.
\begin{equation*}
\chi \sim \frac{C}{(\epsilon-\epsilon_\mathrm{crit})^\gamma} + \frac{C_1}{(\epsilon-\epsilon_\mathrm{crit})^{\gamma_1}} + \frac{C_2}{(\epsilon-\epsilon_\mathrm{crit})^{\gamma_2}} + \ldots \mbox{with $\gamma > \gamma_1 > \gamma_2 \ldots$}
\end{equation*}
where $\epsilon = \epsilon(T)$ is an appropriate normalization of the temperature $T$, and $\epsilon_\mathrm{crit}$ is the normalized critical temperature.. The theory of critical phase transitions is a systematic understanding of these exponents and the relations between them, mainly by employing the apparatus of the renormalization group (see \cite{G18}).

However, this kind of divergence only occurs for systems that are much larger than the size of atoms. For equilibrium systems such as the Ising model, this is reflected in the partition function $Z(L, \beta)$ of the system, where $L$ is the system size and $\beta$ is the inverse temperature. For any finite $L$, the partition function is a smooth function of $\beta$, and there are no singularities, hence no phase transitions. In practice, if the system has a large but finite size $L$, the singularities are ``rounded off' by an amount that decreases as $L$ becomes larger, as illustrated on the left side of Figure \ref{fig:fss2}.

\begin{figure}
\includegraphics[scale=0.4]{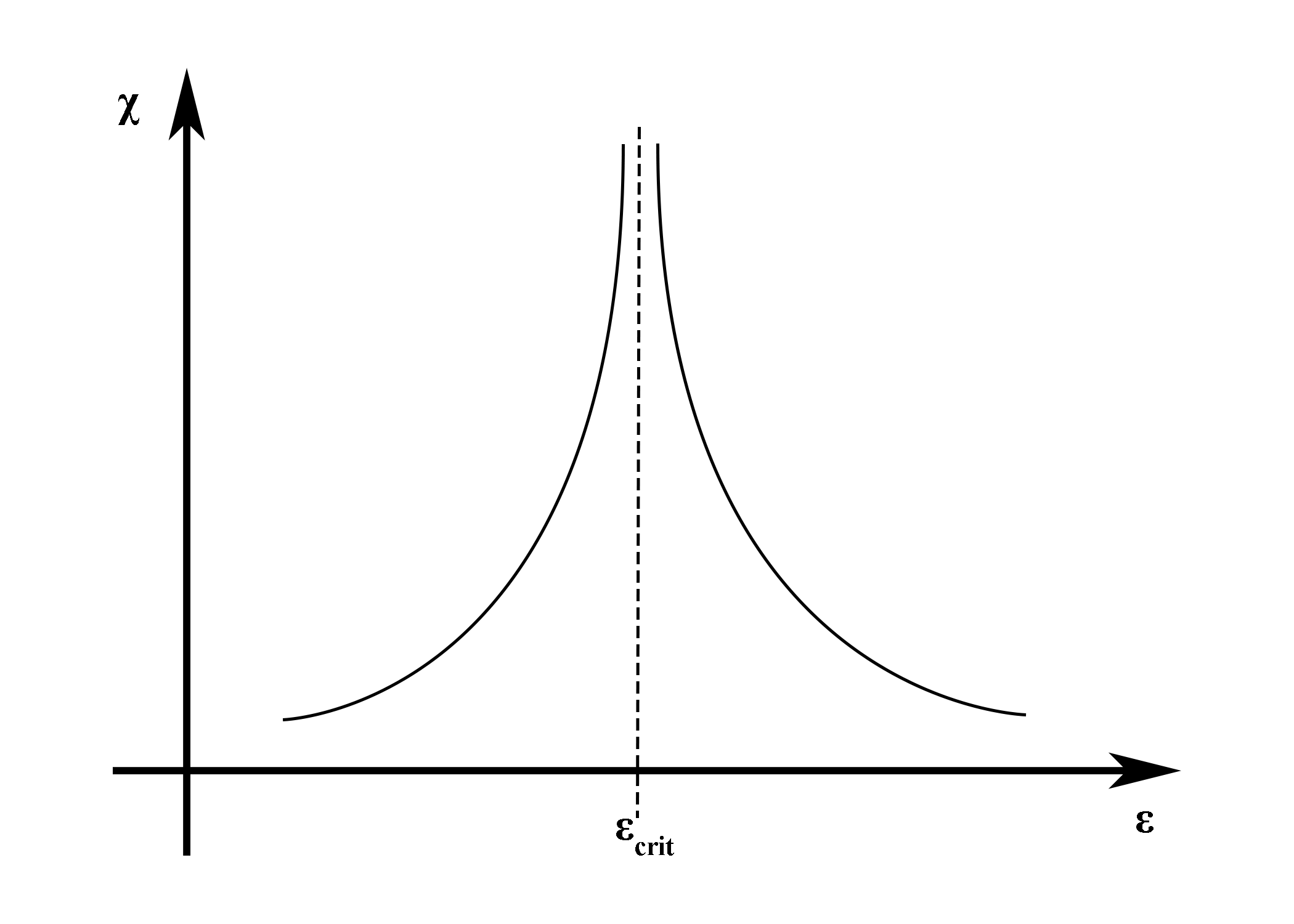}
\caption{$\chi$ (when $L = \infty$) as a function of normalized temperature $\epsilon$, diverging near $\epsilon_\mathrm{crit}$.} \label{fig:fss1}
\end{figure}

\begin{figure}
\includegraphics[scale=0.4]{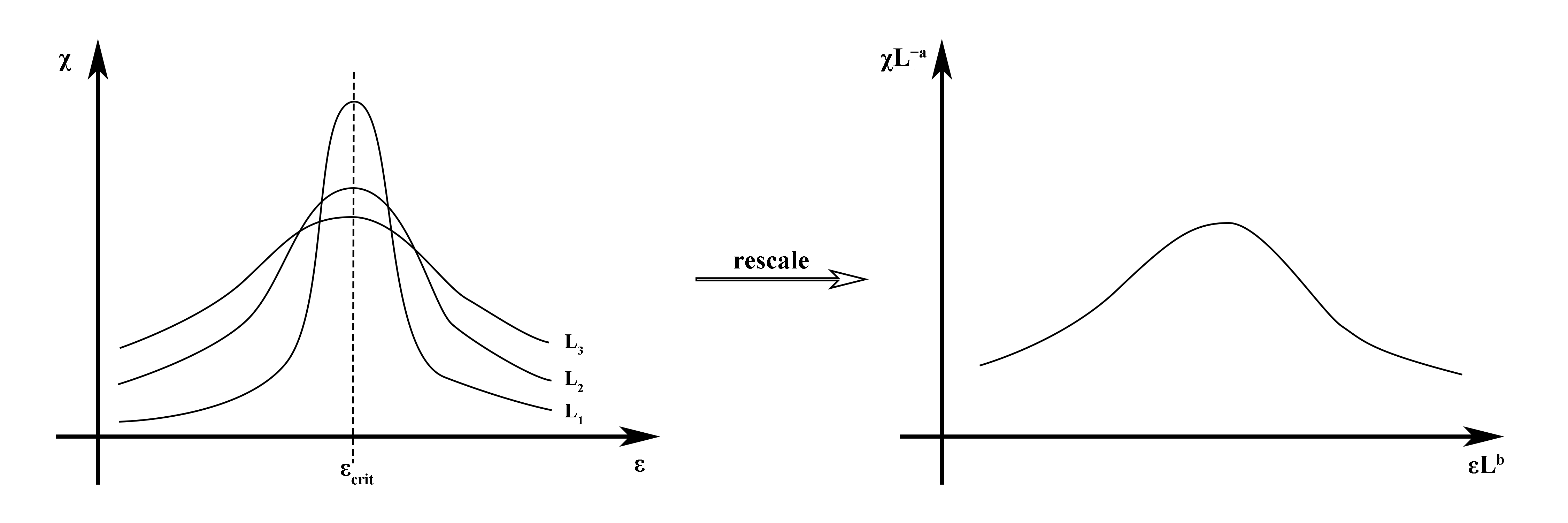}
\caption{Left: $\chi(L, \epsilon)$ for different system sizes $L_1 > L_2 > L_3$. Right: upon a suitable scaling of the coordinates, $\chi$ becomes nearly identical for any $L$.} \label{fig:fss2}
\end{figure}

Remarkably, it is found that these curves for $\chi(L, T)$ for different $L$ near the critical point can be made to collapse on each other, by scaling both $x$ and $y$-axes by factors depending on $L$, so that one has
\begin{equation*}
\chi(L,T) \sim L^af((\epsilon-\epsilon_\mathrm{crit}) L^b),
\end{equation*}
for some function $f$ and constants $a,b$ --- see Figure \ref{fig:fss2}. This scaling collapse is called the \emph{finite size scaling}. In addition, for each $\epsilon$ away from $\epsilon_\mathrm{crit}$, $\chi$ converges to a finite value as $L \rightarrow \infty$; from this it must be that
\begin{equation*}
f(x) \sim \frac{1}{x^{a/b}} \mbox{ for $x$ near $\infty$.}
\end{equation*}
Hence, for each $\epsilon \neq \epsilon_\mathrm{crit}$, $\chi \sim (\epsilon-\epsilon_\mathrm{crit})^{-a/b}$ as $L \rightarrow \infty$. On the other hand, by making $T$ approach the critical temperature at a rate such that $(\epsilon-\epsilon_\mathrm{crit}) L^b$ is large but constant, we obtain $\chi(L, \epsilon_{\mathrm{crit}}) \sim L^a$ for $L \gg 0$. These relations can be used to study $\chi(\infty, T)$ by looking at $\chi(L, T)$ for finite values of $L$, for example.

In non-equilibrium systems such as sandpile models, the temperature is no longer a parameter that an external observer controls; rather, as the dynamics unfolds, the system approaches the critical temperature on its own (hence the term \emph{self-organized criticality} (SOC) systems, as they are sometimes called). Therefore, the above story needs some tweaking, but similar statements hold. For sandpile models, one interprets $\epsilon = z_L$ and $\epsilon_\mathrm{crit} = z_c$, where $z_L = \mathbb{E}(z(r))$ is the average of $z(r):=(1/L)\sum r(i)$ taken over the steady state of size $L$ system, and $z_c = \lim_{L \rightarrow \infty} z_L$ is the critical ``temperature.'' Then one has the relation
\begin{equation} \label{eq:1st}
z_c = z_L + \frac{C}{L^{\sigma}} + \mbox{(smaller errors)}
\end{equation}
for some constants $C$ and $\sigma$, akin to what one would obtain by putting together the two relations $\chi \sim (\epsilon-\epsilon_\mathrm{crit})^{-a/b}$ and $\chi \sim L^a$ discussed earlier. Moreover, FSS also predicts that
\begin{equation} \label{eq:2nd}
\mathrm{Var}(z(r)) \sim L^{-2\sigma}
\end{equation}
with the same $\sigma$. In the literature, for each system, the letter $\sigma$ is reserved to denote the constant such that \eqref{eq:1st} or \eqref{eq:2nd} holds.

There also exist the FSS theory of boundary behavior --- see e.g. Diehl \cite{D97}. In the case of the Ising model, write $m(T)$ for the bulk magnetization at temperature $T$, and $m(i,T)$ for the mean magnetization at distance $i$ from the surface. Then, for the system size $L \gg 0$, there is a relation
\begin{equation*}
m(T) - m(i,T) \sim i^{-a}f((\epsilon - \epsilon_\mathrm{crit})^b i)
\end{equation*}
for some exponents $a$ and $b$, where $f(x) \sim \exp(-cx)$ for a constant $c>0$ and $x$ large. Similarly, for sandpile models, the average of the $i$-th pile $r(i)$ satisfies
\begin{equation} \label{eq:3rd}
z_c - \mathbb{E}(r(i)) \sim i^{-a_1} \mbox{ or } (L+1-i)^{-a_2},
\end{equation}
for some $a_1$ and $a_2$, depending on whether $i$ is closer to $1$ or $L$. For Abelian models, thanks to its inherent left-right symmetry, it can be argued theoretically and experimentally that $a_1 = a_2 = \sigma$. For non-Abelian models, it is possible that $a_1 \neq a_2$.

Recall that the root Hermite factor (RHF) of a configuration $r$ is defined as
\begin{equation}\label{eq:logrhf}
\log \mbox{RHF$(r)$} = \frac{1}{(L+1)^2}\sum_{i=1}^L (L+1-i)r(i).
\end{equation}

Write $y_L$ for the (empirical) average of the RHF of LLL in dimension $n = L+1$, and $y_c = \lim_{L \rightarrow \infty} y_L$. The analogous statements to \eqref{eq:1st} and \eqref{eq:2nd} for RHF then becomes
\begin{align}
& y_c = y_L + \frac{D}{L^\sigma} + \mbox{(smaller errors)} \tag{\ref{eq:1st}'} \label{eq:1st'} \\
& \mathrm{Var}(y_L) \sim L^{-2\sigma}. \tag{\ref{eq:2nd}'} \label{eq:2nd'}
\end{align}

\subsection{Design}

We ran extensive experiments on dimensions $100, 150, 200, 250, 300$, with at least 50,000 iterations for each dimension, to test the formulas \eqref{eq:1st}, \eqref{eq:1st'}, \eqref{eq:2nd}, \eqref{eq:2nd'}, \eqref{eq:3rd} on the LLL algorithm. It was quite a sizable experiment, involving more than 300 cores for over four months. Unlike in the previous sections, we use the original LLL here, with $\delta = 0.999$.

We tried a couple of different methods to generate random bases: the same method as in Section 2 above, with determinant $2^{10n}$ and also with determinant $2^{5n}$, and the knapsack-type bases. We found that they all yield the same results in the lower dimensions, so for dimensions $\geq 200$ we only used the knapsack-type bases with parameter $20n$, which are $n \times (n+1)$ matrices of form
\begin{equation*}
\begin{pmatrix}
x_1 & 1 & & &  \\
x_2 & 0 & 1 & &  \\
\vdots & \vdots & \vdots & \ddots & \\
x_n & 0 & \cdots & 0 & 1
\end{pmatrix}
\end{equation*}
where $x_1, \ldots, x_n$ are integers sampled from $[0, 2^{20n})$ uniformly.

\subsection{Average and variance of RHF}

Table \ref{table:1st} (graphically depicted in Figure \ref{fig:1st}) summarizes our data on the averages of $z_L$ and $y_L$. It demonstrates that our data fits very well with \eqref{eq:1st} and \eqref{eq:1st'}, with $\sigma = 0.75$. Accordingly, we obtain the numerical estimates
\begin{equation} \label{eq:zy_pred1}
z_L \approx 0.0448 - 0.194L^{-3/4}, y_L \approx 0.0224 - 0.09L^{-3/4},
\end{equation}
and thus
\begin{equation} \label{eq:rhf_extrapolate1}
\mathrm{RHF_L} \approx \exp(0.0224 - 0.09L^{-3/4}) \rightarrow 1.02265\ldots \mbox{ as $L \rightarrow \infty$},
\end{equation}
which is close but slightly higher than the ``1.02.''

\begin{table}[]
\begin{tabular}{|l|lllll|}
\hline
dim($=L+1$)       & \multicolumn{1}{c}{100} & \multicolumn{1}{c}{150} & \multicolumn{1}{c}{200} & \multicolumn{1}{c}{250} & \multicolumn{1}{c|}{300} \\ \hline
$z_L$             & 0.03866             & 0.04028             & 0.04115             & 0.04172             & 0.04211              \\ \hline
$z_L - CL^{-\sigma}$ & 0.04479             & 0.04480             & 0.04480             & 0.04480             & 0.04480               \\ \hline
$y_L$             & 0.01957             & 0.02032             & 0.02072              & 0.02098             & 0.02116              \\ \hline
$y_L - DL^{-\sigma}$ & 0.02242             & 0.02242            & 0.02241             & 0.02241             & 0.02240              \\ \hline
\end{tabular}
\caption{Results on $z_L$ and $y_L = \mathbb{E}(\log\mathrm{RHF})$, with $\sigma = 3/4$, $C = -0.194$ and $D = -0.09$.}
\label{table:1st}
\end{table}

\begin{figure}
\includegraphics[scale=0.5]{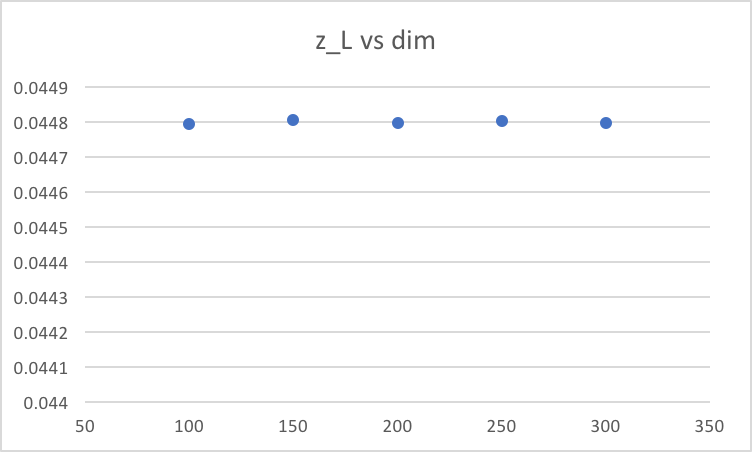}
\includegraphics[scale=0.5]{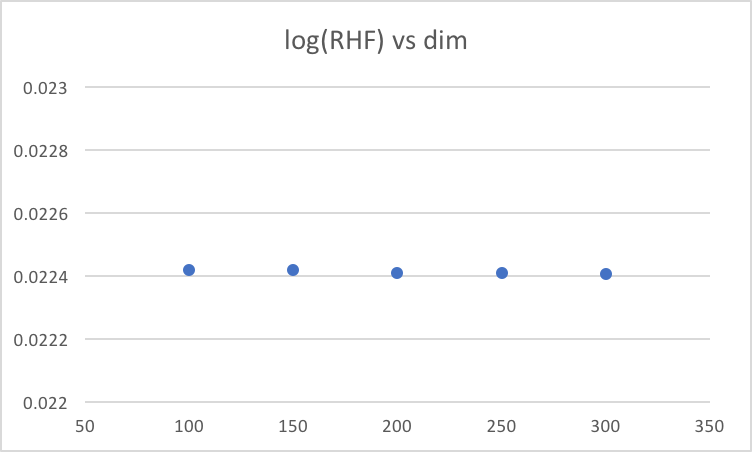}
\caption{Left: dimension versus $z_L - 0.194L^{-3/4}$. Right: dimension versus $y_L - 0.09L^{-3/4}$.}
\label{fig:1st}
\end{figure}

Table \ref{table:2nd} and Figure \ref{fig:2nd} show our data on the variances of $z_L$ and $y_L$. They also fit \eqref{eq:2nd} and \eqref{eq:2nd'} quite well, with the same $\sigma = 0.75$, though to a slightly lesser extent.

\begin{table}[]
\begin{tabular}{|l|lllll|}
\hline
dim($=L+1$)           & \multicolumn{1}{c}{100} & \multicolumn{1}{c}{150} & \multicolumn{1}{c}{200} & \multicolumn{1}{c}{250} & \multicolumn{1}{c|}{300} \\ \hline
$V(z_L)$              & $2.24 \times 10^{-6}$            & $1.21 \times 10^{-6}$             & $7.84\times 10^{-7}$             & $5.62 \times 10^{-7}$             & $4.21\times 10^{-7}$              \\ \hline
$V(z_L)/L^{-2\sigma}$ & 0.00224             & 0.00222             & 0.00222             & 0.00222             & 0.00219               \\ \hline
$V(y_L)$              & $1.05 \times 10^{-6}$              & $5.44\times10^{-7}$             & $3.49\times10^{-7}$             & $2.45\times10^{-7}$             & $1.84\times10^{-7}$              \\ \hline
$V(y_L)/L^{-2\sigma}$ & 0.00105               & 0.00100             & 0.00099             & 0.00097             & 0.00096              \\ \hline
\end{tabular}
\caption{Results on the variances of $z_L$ and $y_L$, with $\sigma = 3/4$.}
\label{table:2nd}
\end{table}

\begin{figure}
\includegraphics[scale=0.5]{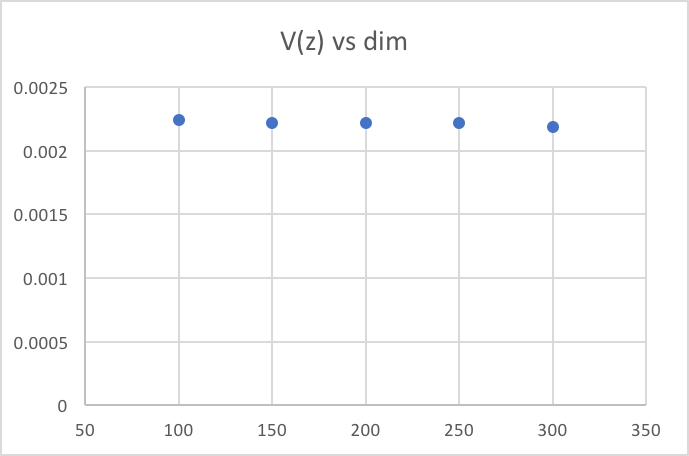}
\includegraphics[scale=0.5]{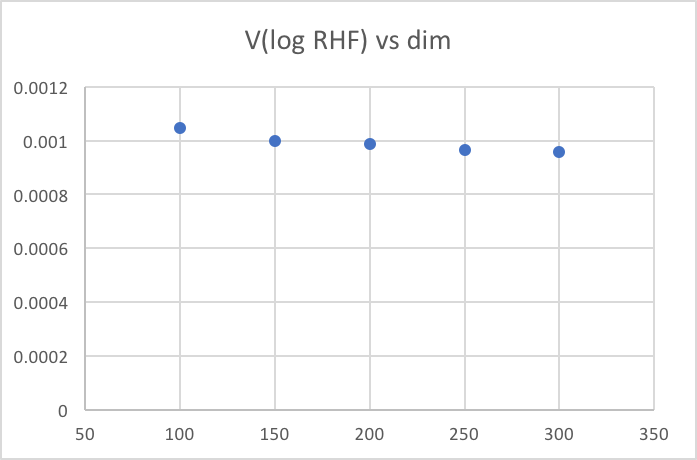}
\caption{Left: dimension versus $V(z_L)/L^{-1.5}$. Right: dimension versus $V(y_L)/L^{-1.5}$.}
\label{fig:2nd}
\end{figure}

\subsection{Boundary statistics}

Figures \ref{fig:3rd_f} and \ref{fig:3rd_b} present comparisons of our data with \eqref{eq:3rd}, with Figure \ref{fig:3rd_f} examining the left boundary (i.e. $i$ near $1$) and Figure \ref{fig:3rd_b} the right boundary (i.e. $i$ near $L$). Here we used $z_c = 0.448$, obtained in Section 3.1 above.

From Figure \ref{fig:3rd_b}, on the right boundary we do find that $z_c -\mathbb{E}(r(L-i)) \sim i^{-0.75}$ on the first $10$ points or so. However, Figure \ref{fig:3rd_f}, and  also the rest of the points on Figure \ref{fig:3rd_b}, makes matters more subtle: it appears that, on the left end, and for many points on the right end, $z_c - \mathbb{E}(r(i)) \sim i^{-1.05}$ appears to be the correct observation.

\begin{figure}
\includegraphics[scale=0.5]{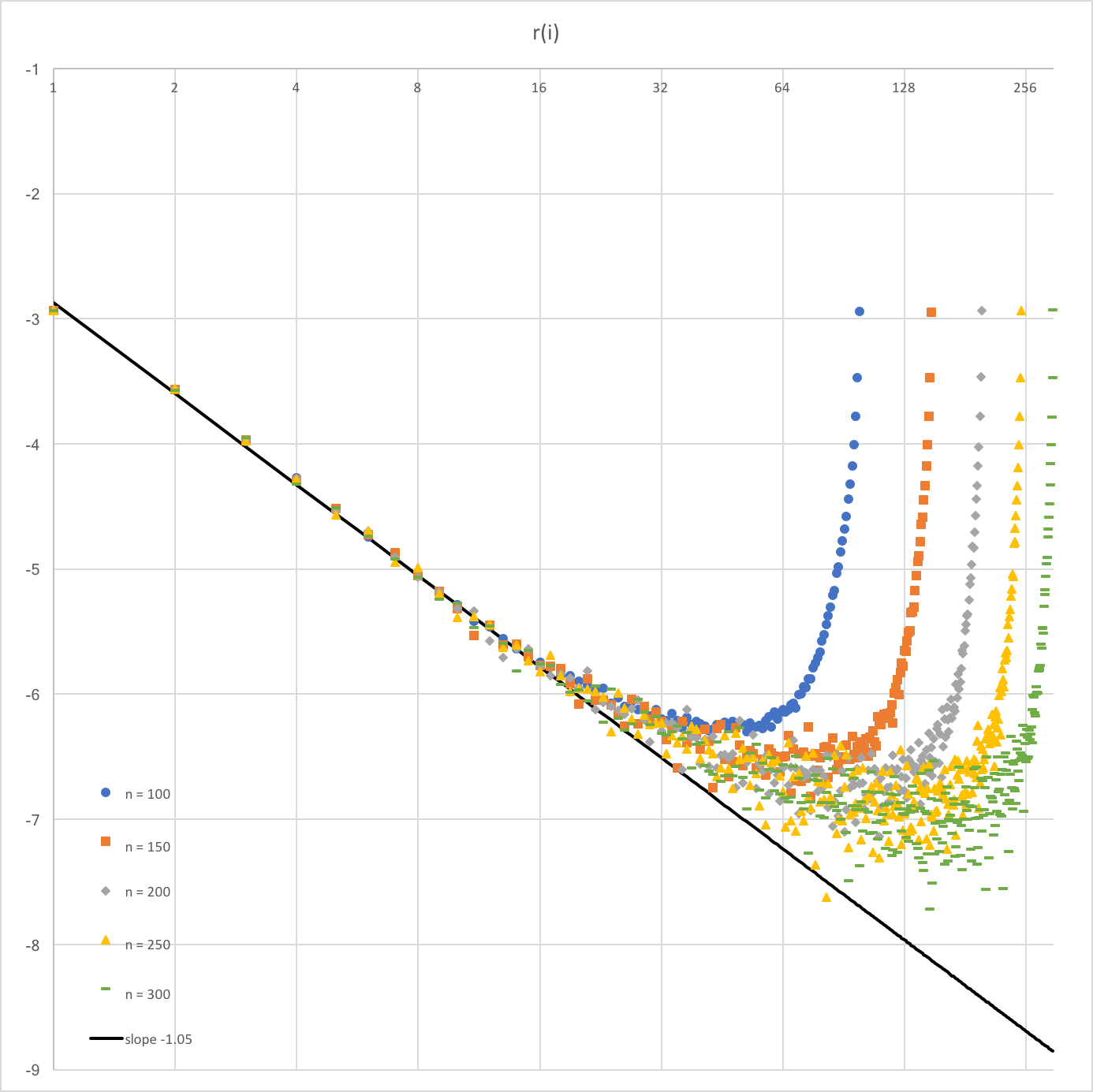}
\caption{$i$ versus $\log(z_c - \mathbb{E}(r(i)))$.}
\label{fig:3rd_f}
\end{figure}

\begin{figure}
\includegraphics[scale=0.5]{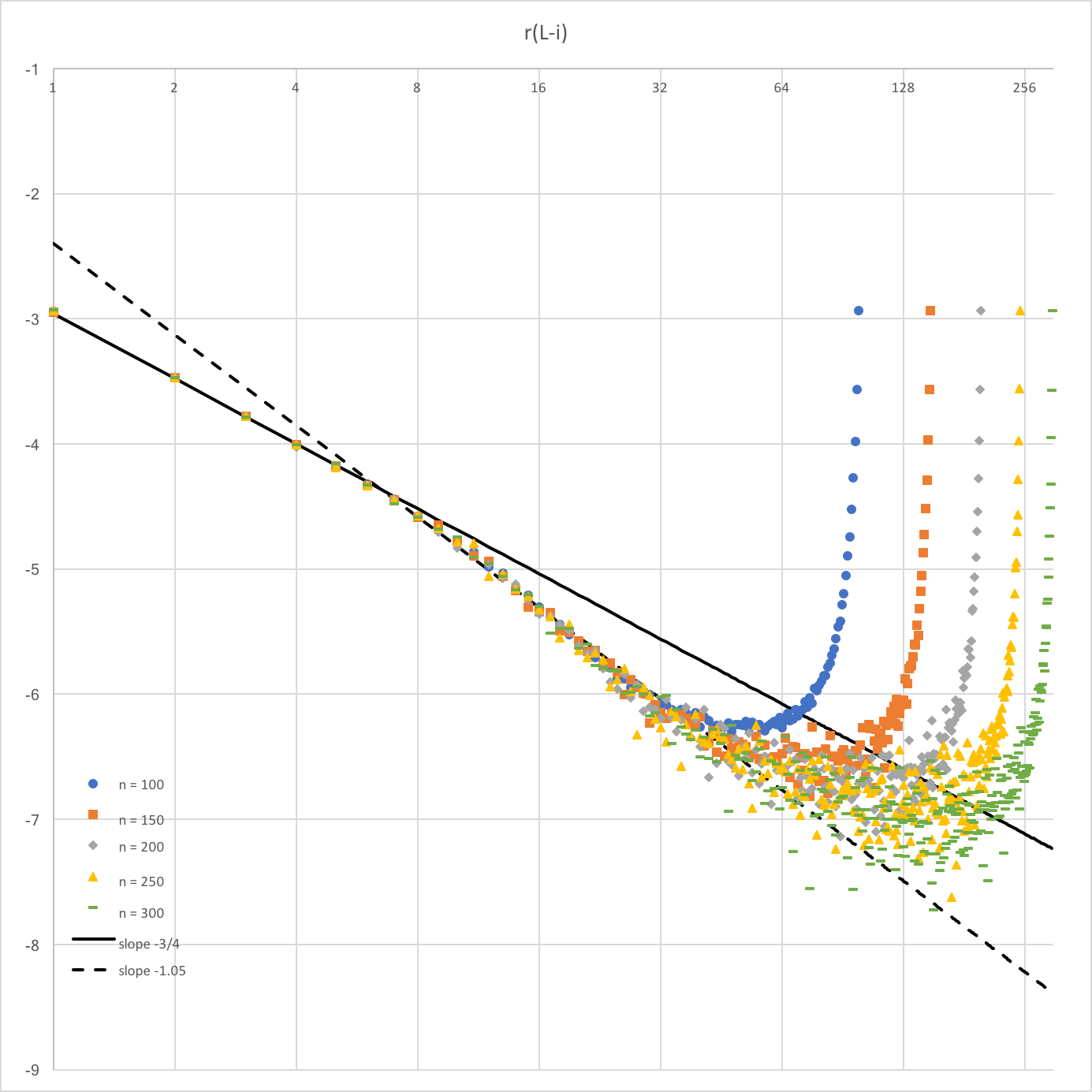}
\caption{$i$ versus $\log(z_c - \mathbb{E}(r(L-i)))$.}
\label{fig:3rd_b}
\end{figure}

\subsection{Summary and discussions}

Typically in physics, experiments of this kind are carried out up to $L$ close to a million, if not more. An experiment of such magnitude is clearly infeasible for lattice reduction, and hence we have been severely constrained in our experiments from the physical perspective. In addition, our estimates of the critical exponent $\sigma$ and other constants very likely leave much room for improvement, by employing more extensive and elaborate numerical techniques. Despite these limitations, our experiments reveal some clear patterns in the empirical output statistics of LLL, robustly described by formulas from statistical mechanics.

We obtain two particularly notable implications. First, the folklore number ``1.02'' is not too far from the LLL behavior in the limit. One could reasonably suspect that the average-worst case RHF gap is only a peculiarity in the low dimensions, and that it would disappear in the dimension limit, citing the result of \cite{KV17} for instance. But we found evidence that the gap is actually a real phenomenon. Second, Figures \ref{fig:3rd_f} and \ref{fig:3rd_b} provide neat formulas for the average output statistics of LLL, via an appropriate normalization of graphs such as Figure \ref{fig:output}. This is a vast refinement of GSA, at least for the LLL algorithm. Of course, the same set of experiments can be carried out for BKZ, and our pilot experiments with BKZ-20 look promising. This result will appear in a forthcoming paper.

It remains a mystery as to how to explain the boundary phenomenon that we observed here. It is not entirely surprising for non-Abelian models to behave differently on the left and right ends, but the particular shape of Figure \ref{fig:3rd_b} is not seen often even in physics, to the best of our knowledge. It is probable that the more familiar pattern may emerge with more data.

\end{document}